\documentclass[A4]{article}
\usepackage{graphicx}
\usepackage{amssymb}
\usepackage{amsmath}
\usepackage{amsthm}
\usepackage{color}
\usepackage{fullpage}

\graphicspath{{figures/}}

\newtheorem{theorem}{Theorem}
\newtheorem{lemma}[theorem]{Lemma}
\newtheorem{proposition}[theorem]{Proposition}
\newtheorem{observation}[theorem]{Observation}
\usepackage{ wasysym }

\usepackage[usenames,dvipsnames]{xcolor}
\newcommand{\figurenames}{\figurename{s}}
\begin{document}

\title{Modem Illumination of Monotone Polygons\thanks{%
A preliminary version of this work has been presented at EuroCG'09~\cite{Aich}.\newline
O. Aichholzer and B. Vogtenhuber partially supported by the ESF EUROCORES programme EuroGIGA -- CRP `ComPoSe', Austrian Science Fund (FWF): I648-N18. %\\
R. Fabila-Monroy partially supported by CONACyT (Mexico), grant 153984. %\\
D. Flores-Pe\~naloza partially supported by CONACyT (Mexico), grant 168277, and PAPIIT IA102513 (UNAM, Mexico). %\\
T. Hackl supported by the Austrian Science Fund (FWF): P23629-N18 `Combinatorial Problems on Geometric Graphs'. % \\
J. Urrutia partially supported by CONACyT (Mexico) grant CB-2007/80268.
  }
 }

\author{Oswin Aichholzer\thanks{Institute for Software Technology, Graz University of Technology, Graz, Austria, \texttt{[oaich|thackl|apilz|bvogt]@ist.tugraz.at}.}
\and Ruy Fabila-Monroy\thanks{Departamento de Matem\'aticas, Cinvestav, D.F. M\'exico, M\'exico, \texttt{ruyfabila@math.cinvestav.edu.mx}.}
\and David Flores-Pe\~naloza\thanks{Departamento de Matem\'aticas, Facultad de Ciencias, Universidad Nacional Aut\'onoma de M\'exico, D.F. M\'exico, M\'exico, \texttt{dflorespenaloza@gmail.com}.}
\and Thomas Hackl$^\dagger$
\and Jorge Urrutia\thanks{Instituto de Matem\'aticas, Universidad Nacional Aut\'onoma de M\'exico, D.F. M\'exico, M\'exico, \texttt{urrutia@matem.unam.mx}.}
\and Birgit Vogtenhuber$^\dagger$
}

\maketitle

%------------------------------ Text -------------------------------------
\begin{abstract}
  We study a generalization of the classical problem of the illumination of polygons.
  Instead of modeling a light source we model a wireless device whose radio signal can penetrate a given number~$k$ of walls.
  We call these objects \mbox{$k$-modems} and study the minimum number of \mbox{$k$-modems} sufficient and sometimes necessary to illuminate monotone and monotone orthogonal polygons.
  We show that every monotone polygon with~$n$ vertices can be illuminated with $\big\lceil \frac{n-2}{2k+3} \big\rceil$ \mbox{$k$-modems}.
  In addition, we exhibit examples of monotone polygons requiring at least $\lceil \frac {n-2} {2k+3}\rceil$ \mbox{$k$-modems} to be illuminated.

  For monotone orthogonal polygons with~$n$ vertices we show that for~$k=1$ and for even~$k$, every such polygon can be illuminated with $\big\lceil \frac{n-2}{2k+4} \big\rceil$ \mbox{$k$-modems}, while for odd~$k\geq3$, $\big\lceil \frac{n-2}{2k+6} \big\rceil$ $k$-modems are always sufficient.
  Further, by presenting according examples of monotone orthogonal polygons, we show that both bounds are tight.
\end{abstract}

%  \keywords{Art Gallery Problem \and Modem Illumination \and $k$-transmitter \and $k$-Modem}

\section{Introduction}

New technologies inspire new research problems, and wireless networking is a typical example of this.
Nowadays, wireless technologies surround us everywhere.
We use them in devices such as cellular phones, satellite communications, and, in our homes, we use wireless modems to connect to the Internet.

This has triggered, among other things, the development of a new class of algorithms designed specifically to work with wireless networks, such as cellular networks, sensor networks, and \emph{ad-hoc} networks~\cite{Kran2,Urrutia2,Watten}.
The development of GPS, also a byproduct of wireless technologies, has allowed the development of so-called ``local algorithms" for routing problems in cellular and \emph{ad-hoc} networks~\cite{BMSU,Krana,Urrutia2} that allow relayed communication between any two nodes $u$~and~$v$ of a network, at any time using  only the position of $u$~and~$v$, as well as the current position of a message while traveling from~$u$~to~$v$.
For more details see~\cite{Urrutia2,Watten}.

In this paper we study what we call the \emph{Modem Illumination Problem}.
This problem stems from our daily use of laptop computers and wireless modems.
Experience shows that when trying to connect a laptop to a wireless modem, there are two factors that have to be considered: the \emph{distance} to the wireless modem and, perhaps most important in many buildings, the \emph{number of walls} separating our laptop from the wireless modem.
From now on, the term \emph{modem} will be used to refer to a wireless modem.
We call a modem \emph{a $k$-modem} if it is strong enough to transmit a stable signal through~$k$ walls along a straight line.
Thus, we say that a point~$p$ in a polygon~$P$ is illuminated by a \mbox{$k$-modem}~$M$ in~$P$ if the line segment joining~$p$ to~$M$ \emph{crosses} at most~$k$ walls (edges) of~$P$.\\

\noindent \emph{\bf{The Modem Illumination Problem}: }
Let~$P$ be a student center (f.k.a. art gallery) modeled by a polygon~$P$ with~$n$ vertices.
How many \mbox{$k$-modems} located at points of~$P$ are always sufficient, and sometimes necessary, to illuminate all points in~$P$?\\

We point out that we allow a modem to be located at a point~$q$ on an edge~$e$ (or even on a vertex~$v$) of~$P$.
In this case, we do not consider~$e$ (or the two edges incident to~$v$) as a barrier for the modem.
That is, the line segments connecting~$q$ and any point of $P$ do not cross~$e$ (or the two edges incident to~$v$).
Thus, if~$p$ is an interior point of~$P$, the line segment connecting~$p$ and~$q$ may cross an odd number of edges of $P$.

For $k=0$ our problem corresponds to Chv\'atal's Art Gallery Theorem~\cite{Chvatal} which states that $\big\lfloor\frac{n}{3}\big\rfloor$ watchmen are always sufficient and sometimes necessary to guard an art gallery with $n$ walls.
Many generalizations of the original Art Gallery problem have been studied, see~\cite{ORourke,Shermer,Urrutia} for comprehensive surveys.
The modem illumination problem was introduced in~\cite{Aich,Ruy}.

A similar problem, posed by Urrutia and solved by Fulek et al.~\cite{Pach}, is the following:
What is the smallest number $\tau = \tau (n)$ such that in any collection of $n$ pairwise disjoint convex sets in the $d$-dimensional Euclidean space, there is a point such that any ray emanating from it meets at most $\tau$ sets of the collection?
In our language, they proved that one \mbox{$\frac{dn+1}{d+1}$-modem} is always sufficient to illuminate the $d$-dimensional Euclidean space in presence of~$n$ convex obstacles.

In~\cite{Ball}, several variations of the Modem Illumination Problem are studied.
For example they present upper and lower bounds on the number of \mbox{$k$-modems} needed to illuminate the plane in the presence of obstacles modeled by line segments (with fixed slopes), or nested polygons.
They also present some bounds on illuminating special classes of simple polygons with \mbox{$2$-modems} in the interior of the polygon.

Polygon illumination with wireless devices has also been studied in a slightly different context,  the so-called sculpture garden problem~\cite{EGS,CHOU}.  
There, each device only broadcasts a signal within a given angle of the polygon and has unbounded range.
The task is to describe the polygon (that is, distinguish it from the exterior) by a combination of the devices, meaning that for each point~$p$ in the interior of the polygon no point outside the polygon receives signals from the same devices as~$p$. 

We remark that the general case of the Modem Illumination Problem for \mbox{$k$-modems} is widely open.
Almost no tight bounds are known so far for any class of polygons for any $k \geq 1$; 
see for example the recent Column of O'Rourke~\cite{ORourke2}.
Very recently, there has been some development on these problems in several directions.
Duque and Hidalgo~\cite{DH} have announced
an upper bound of $O(\frac{n}{k})$ on the number of  \mbox{$k$-modems}
needed to illuminate the interior of a simple polygon of $n$ sides.
In the case of orthogonal polygons they give a tighter bound of $6\frac{n}{k}+1$.
For the related question on ``edge-transmitters'' (modems where the signal is emanated from a whole edge instead of only a point; see also~\cite{Shermer}), Cannon et al.~\cite{edgetrans} consider several classes of polygons.
They show lower (i.e., sometimes necessary) and upper (i.e., always sufficient) bounds on the number of needed transmitters for simple (orthogonal) (monotone) polygons.
For example, for monotone polygons they prove $\big\lceil\frac{n-3}{8}\big\rceil$ and $\big\lceil\frac{n-2}{9}\big\rceil$ as lower and upper bounds, and they give a tight bound of $\big\lceil\frac{n-2}{10}\big\rceil$ for monotone orthogonal polygons.
Finally, in~\cite{np}, the same authors study the corresponding algorithmic optimization problems for both, point and edge transmitters.
In particular, they show that it is NP-hard to compute the minimum number of \mbox{$k$-modems} needed to illuminate a given simple polygon. 

% cite {edgetrans}, bounds for 2-edge transmitters (if we want any of them): 
% $\leq \lfloor\frac{3n}{10}\rfloor, \geq \lfloor\frac{n}{6}\rfloor$ for general polygons (upper not by them), 
% $\leq \lceil\frac{n-3}{8}\rceil, \geq \lceil\frac{n-2}{9}\rceil$ for monotone, 
% $= \lceil\frac{n-2}{10}\rceil$ for monotone orthogonal, 
% $\leq \lfloor\frac{3n+4}{16}\rfloor, \geq \lceil\frac{n-2}{10}\rceil$ for orthogonal (upper by {Edgeguards}).} 

In this paper we provide lower and upper bounds for the Modem Illumination Problem for monotone polygons and monotone orthogonal polygons, which are a reasonable model of most real life buildings.
These bounds are asymptotically tight (i.e., the lower and upper bounds match):
$\big\lceil \frac{n-2}{2k+3} \big\rceil$ \mbox{$k$-modems} for monotone polygons;
$\big\lceil \frac{n-2}{2k+4} \big\rceil$ \mbox{$k$-modems} for monotone orthogonal polygons for~$k=1$ and even~$k$;
and $\big\lceil \frac{n-2}{2k+6} \big\rceil$ \mbox{$k$-modems} for monotone orthogonal polygons for~odd~$k\geq3$.
(Note that these bounds improve the bounds from the preliminary version~\cite{Aich}.)

\section{Illumination of (general) monotone polygons with $k$-modems}
\label{sec:gmono}

To keep things as simple as possible we introduce several conventions on monotone polygons.
When we speak of a polygon, we refer to both the boundary and the interior of the polygon.
For technical reasons and without loss of generality, we make the following assumption: for a non-ortho\-gonal monotone polygon, we assume that no two of its edges (on different chains) are parallel.

For any monotone polygon, we assume that the direction of monotonicity is the \mbox{$x$-axis}, and that no two vertices have the same $x$-coordinate.
We denote a (monotone) polygon~$P$ with~$n$ vertices as (monotone) \mbox{$n$-gon}.
Further, we denote the vertices of~$P$ by $v_1, \ldots, v_{n}$, where the labels are given to the vertices with respect to their $x$-sorted order, such that $v_1$ is the leftmost and $v_{n}$ is the rightmost vertex of~$P$.
Note that this implies that, in general, the vertices are not labeled with respect to their order along the boundary of~$P$.

To validate possible \mbox{$k$-modem} positions we will consider rays from these positions.
For a point $q$ of $P$ let ${\cal R}(q)$ be the set of all rays starting at $q$.
Observe that at the last edge~$e$ of~$P$ that is intersected by a ray~$r$, $r$ leaves~$P$ and never enters~$P$ again.
In other words, no line segment connecting~$q$ with a point~$p$ on~$r\cap P$ is crossing~$e$.
Hence, a \mbox{$k$-modem} does not have to overcome $e$ in the direction of $r$ to fully illuminate $P$.
If a \mbox{$k$-modem} at $q$ illuminates~$P$ (in the direction of a ray $r$), then we say that~$q$ is a \emph{valid} \mbox{$k$-modem} position for~$P$ (in the direction of~$r$).

\begin{observation}\label{obs:rayk+1}
For any point $q$ of $P$ and each $r\in{\cal R}(q)$, $q$ is a valid \mbox{$k$-modem} position for $P$ in the direction of $r$ if and only if $r$ intersects at most $k+1$ edges of~$P$.

A point~$q$ is a valid \mbox{$k$-modem} position for~$P$ if and only if every $r\in{\cal R}(q)$ intersects at most $k+1$ edges of~$P$.
\end{observation}

If $q$ is a point on the boundary~$\partial P$ of~$P$ then we distinguish two subsets of rays.
Let ${\cal R}^o(q)\subset {\cal R}(q)$ be the subset of rays that start to the outside of~$P$ and let ${\cal R}^i(q)\subset {\cal R}(q)$ be the subset of rays that start to the inside of~$P$.

\begin{observation}\label{obs:rayio}
For any point~$q$ of~$\partial P$, each ray in ${\cal R}^o(q)$ intersects an even number of edges of~$P$ and each ray ${\cal R}^i(q)$ intersects an odd number of edges of~$P$.
\end{observation}

Recall that, if a possible \mbox{$k$-modem} position~$q$ is on an edge (or vertex) of~$P$, then no ray in ${\cal R}(q)$ crosses that edge (or the two edges incident to that vertex).
With this in mind we state a few simple observations on the size of polygons that can be fully illuminated by a single \mbox{$k$-modem}.

\begin{proposition}\label{prop:anywhere}
Every $(k+2)$-gon~$P$ can be illuminated with a \mbox{$k$-modem} placed anywhere in the interior or on the boundary of~$P$.
\end{proposition}
\begin{proof}
  Let $q$ be the point of~$P$ where the \mbox{$k$-modem} is placed.
  If~$q$ is on the boundary of~$P$, then at least one edge of~$P$ is not crossed by any ray in ${\cal R}(q)$.
  If~$q$ is in the interior of~$P$, then for each $r\in{\cal R}(q)$ at least one edge of~$P$ is crossed by the ray starting in the opposite direction of $r$ from $q$.
  Therefore, each ray in ${\cal R}(q)$ crosses at most $k+1$ edges out of the $k+2$ edges of~$P$.
  By Observation~\ref{obs:rayk+1}, a \mbox{$k$-modem} at~$q$ illuminates the whole polygon. %\qed
\end{proof}

Note that this proposition will not be used in this paper.
It nevertheless is of interest, as it is also true for general (not necessarily monotone) simple polygons.
The next lemma is true for general simple polygons too, and will be needed later on.

\begin{lemma}\label{lem:vertex}
Every $(k+3)$-gon~$P$ can be illuminated with a \mbox{$k$-modem} placed on any vertex of~$P$.
\end{lemma}
\begin{proof}
  Let $q$ be the vertex of~$P$ where the \mbox{$k$-modem} is placed.
  No ray in ${\cal R}(q)$ crosses any of the two edges incident to~$q$.
  Therefore, each ray in ${\cal R}(q)$ crosses at most $k+1$ edges out of the $k+3$ edges of~$P$.
  By Observation~\ref{obs:rayk+1}, a \mbox{$k$-modem} at~$q$ illuminates the whole polygon. %\qed
\end{proof}

For every monotone ($x$-monotone by convention) polygon~$P$, every vertical line~$\ell$ cuts~$P$ into at most two parts, because $P\cap\ell$ is either one connected component or empty.
Let~$H_L(\ell)$ be the (closed) half-plane bounded to the right by~$\ell$.
Following suit, let~$H_R(\ell)$ be the (closed) half-plane bounded to the left by~$\ell$.
Note that $H_L(\ell)$ (or $H_R(\ell)$) is to the left (or right) of~$\ell$ (including~$\ell$).

Hence, for a vertical line~$\ell$ that intersects~$P$, we call $P\cap H_L(\ell)$ and $P\cap H_R(\ell)$ the left and right part of~$P$, respectively.
We say that $P\cap H_L(\ell)$ (or $P\cap H_R(\ell)$) contains an edge~$e$ of~$P$ if $(P\cap H_L(\ell))\cap e$ (or $(P\cap H_L(\ell))\cap e$) is neither empty nor a single point.
%
%(The latter implies that an end point of~$e$, which is a vertex of~$P$, is on~$\ell$. )
%
Note that at least one edge of~$P$ is always split by a vertical line that intersects~$P$, because, by convention, no two vertices of~$P$ have the same $x$-coordinate.

Strictly speaking, $P\cap H_L(\ell)$ and $P\cap H_R(\ell)$ are not polygons.
Their rightmost and leftmost edge, respectively, is missing.
We could fix this by adding the line segment $s=P\cap\ell$ to both parts of~$P$ as an (auxiliary) edge.
But this would violate the convention that no two vertices share the same $x$-coordinate.
To maintain this convention we can slightly perturb the end point of~$s$ that is no vertex of~$P$ along its edge of~$P$ without changing the situation.
Furthermore, we say that $P\cap H_L(\ell)$ (or $P\cap H_R(\ell)$) is illuminated if the polygon $(P\cap H_L(\ell))\cup s$ (or $(P\cap H_R(\ell))\cup s$) would be illuminated.
The following statement summarizes observations about splitting~$P$. % for later use.

\begin{observation}\label{obs:simplesplit}
Let~$P$ be an $x$-monotone $n$-gon and let $\ell_i$ be a vertical line through $v_i$, $2\leq i\leq n-1$.
\begin{itemize}\vspace*{-0.5ex}
\item $P\cap H_L(\ell_i)$ contains $i$ edges of~$P$.
\item $P\cap H_R(\ell_i)$ contains $n-i+1$ edges of~$P$.
\item If both $P\cap H_L(\ell_i)$ and $P\cap H_R(\ell_i)$ are illuminated, then~$P$ is illuminated.
\end{itemize}
\end{observation}

Combining this simple splitting with Lemma~\ref{lem:vertex}, we can state the following:

\begin{proposition}\label{prop:2k+3}
Every $x$-monotone $(2k+3)$-gon~$P$ can be illuminated with a \mbox{$k$-modem} placed on $v_{k+2}$.
\end{proposition}
\begin{proof}
  Let $m=k+2$, let $\ell_{m}$ be the vertical line through~$v_m$, and let~$s_m = P\cap\ell_{m}$.
  The polygon $(P\cap H_L(\ell_m))\cup s_m$ contains $k+2$ edges of~$P$ (Observation~\ref{obs:simplesplit}) plus the (auxiliary) edge~$s_m$.\footnote{%
  More exactly, the boundary of the polygon $(P\cap H_L(\ell_m))\cup s_m$ consists of the first $k+1$ edges of $P$ plus the part of the $(k+2)^{nd}$ edge of $P$ that is to the left of $\ell_m$ plus $s_m$.}
  The polygon $(P\cap H_R(\ell_m))\cup s_m$ contains $n-m+1=k+2$ edges of~$P$ (Observation~\ref{obs:simplesplit}) plus the (auxiliary) edge~$s_m$.
  By Lemma~\ref{lem:vertex}, $v_m$ is a valid \mbox{$k$-modem} position for both $(k+3)$-gons.
  Hence, by Observation~\ref{obs:simplesplit}, a \mbox{$k$-modem} at~$v_m=v_{k+2}$ illuminates the whole polygon~$P$. 
\end{proof}

Combining Observation~\ref{obs:simplesplit} and Proposition~\ref{prop:2k+3} we can derive a first bound for the number of \mbox{$k$-modems}.
It is easy to see that we can split each $x$-monotone $n$-gon into a (left) $x$-monotone $(2k+3)$-gon and a (right) $x$-monotone $(n-2k)$-gon using a vertical line~$\ell_{2k+2}$ through $v_{2k+2}$ (and the auxiliary edge~$s_{2k+2}$).
Recursively splitting the (right) $x$-monotone $(n-2k)$-gon directly leads to a (non-optimal) upper bound of $\big\lceil \frac{n-3}{2k} \big\rceil$ \mbox{$k$-modems} to illuminate a monotone $n$-gon.

We improve this bound in the remainder of this section, which is divided into three steps.
First we increase the size of a polygon that we can guarantee to illuminate with a single \mbox{$k$-modem}.
Then we introduce an efficient way to split a large polygon into smaller subpolygons.
And in the final step we combine the two previous steps to an upper bound on the Modem Illumination Problem and provide matching lower bound examples.

\subsection{Illuminating monotone polygons with a single $k$-modem}
\label{sec:gmono-single}

Before we can improve Proposition~\ref{prop:2k+3}, we need a solid basis of naming conventions.
Similar as before, let $v_m\in\{v_3,\ldots,v_{n-2}\}$ be a vertex of a monotone $n$-gon~$P$, let $\ell_m$ be the vertical line through~$v_m$, and let $f$ be the edge of~$P$ crossed by~$\ell_m$, with~$p_m$ being the point of this intersection.
Let $v^{f}_{L}$ and $v^{f}_{R}$ be the left and right end point of~$f$, respectively.
Let $f_L$ be the edge sharing $v^{f}_{L}$ with $f$ and let $f_R$ be the edge sharing $v^{f}_{R}$ with $f$.
Let $v^e_L$ and $v^e_R$ be the left and right neighbor of $v_m$ on $P$, respectively.
Let $e_L=v_m v^e_L$ and $e_R=v_m v^e_R$.
Observe that these naming conventions induce an ``$f$-side'' and an ``$e$-side'' of~$P$.

\begin{figure}[htb]
  \centering
  \includegraphics{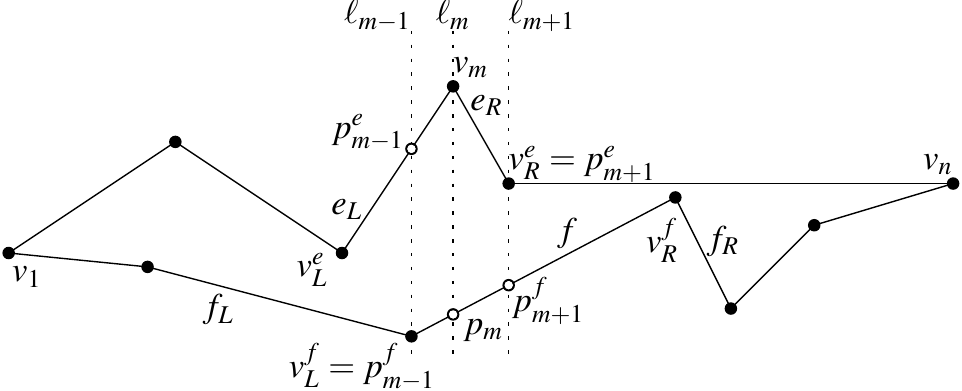}
  \caption{Example to illustrate the naming conventions for a monotone $n$-gon.}
  \label{fig:naming_convention}
\end{figure}

Let $\ell_{m-1}$ and $\ell_{m+1}$ be the vertical lines through $v_{m-1}$ and $v_{m+1}$, respectively.
Let $p^e_{m-1}$ and $p^f_{m-1}$ be the intersection of $\ell_{m-1}$ with $e_L$ and $f$, respectively.
Similarly, let $p^e_{m+1}$ and $p^f_{m+1}$ be the intersection of $\ell_{m+1}$ with~$e_R$ and~$f$, respectively.
Note that either $p^e_{m-1}$ or $p^f_{m-1}$ is $v_{m-1}$ and either~$p^e_{m+1}$ or~$p^f_{m+1}$ is~$v_{m+1}$.
Thus, for each of $\ell_{m-1}$ and $\ell_{m+1}$, one of the intersections is not a proper crossing with the respective edge, because the line passes through a vertex of~$P$.
Further note that $v_{m-1}$ is either $v^e_L$ or $v^{f}_{L}$, and that $v_{m+1}$ is either $v^e_R$ or $v^{f}_{R}$.

For simplicity and without loss of generality, let~$v_m$ be above~$f$.
Hence, the ``$e$-side'' of~$P$ is the upper side of~$P$ and the ``$f$-side'' is the lower one.
With upper and lower side of~$P$ we name the respective pieces of the boundary of~$P$ that result from removing the vertices $v_1$ and~$v_n$.
See \figurename~\ref{fig:naming_convention} for an illustrating example.

Let~$P$ be an $x$-monotone $(2k+5)$-gon and let $m=\frac{2k+5+1}{2}=k+3$.
We show that at least one out of $p^e_{m-1}$, $p^f_{m-1}$, $v_{m}$, $p_{m}$, $p^e_{m+1}$, and $p^f_{m+1}$ is a valid \mbox{$k$-modem} position for~$P$.
We distinguish the two cases of $k$ being even and $k$ being odd, because the proofs differ slightly for these cases.

\begin{lemma}\label{lem:2k+5even}
Let $k\geq0$ be even.
Every $x$-monotone $(2k+5)$-gon~$P$ can be fully illuminated by a \mbox{$k$-modem} placed on at least one position out of $v_m$, $p_m$, $p^f_{m-1}$, and $p^f_{m+1}$.
\end{lemma}

\begin{lemma}\label{lem:2k+5odd}
Let $k\geq1$ be odd.
Every $x$-monotone $(2k+5)$-gon~$P$ can be fully illuminated by a \mbox{$k$-modem} placed on at least one position out of $v_m$, $p_m$, $p^e_{m-1}$, $p^f_{m-1}$, $p^e_{m+1}$, and $p^f_{m+1}$.
\end{lemma}

We postpone the proofs of both lemmas until we have proven the necessary tools.
The following theorem simply summarizes the two lemmas and realizes the desired improvement over Proposition~\ref{prop:2k+3}.

\begin{theorem}\label{thm:2k+5}
For $k\geq0$ and $m=k+3$, every $x$-monotone $(2k+5)$-gon~$P$ can be illuminated by a \mbox{$k$-modem} placed on at least one position out of $v_m$, $p_m$, $p^e_{m-1}$, $p^f_{m-1}$, $p^e_{m+1}$, and $p^f_{m+1}$.
\end{theorem}

In the proof of Proposition~\ref{prop:2k+3} we split the polygon in the middle.
To prove Lemmas~\ref{lem:2k+5even} and~\ref{lem:2k+5odd} (and therefore Theorem~\ref{thm:2k+5}) we consider three possible splitting lines, namely $\ell_{m-1}$, $\ell_m$, and $\ell_{m+1}$.
We will start by proving always valid \mbox{$k$-modem} positions for the leftmost and rightmost parts of~$P$.

\begin{lemma}\label{lem:2k+5always}
Let $k\geq0$ and let~$P$ be an $x$-monotone $(2k+5)$-gon.
The points $p^e_{m-1}$ and $p^f_{m-1}$ are valid \mbox{$k$-modem} positions for $P\cap H_L(\ell_{m-1})$, and $p^e_{m+1}$ and $p^f_{m+1}$ are valid \mbox{$k$-modem} positions for $P\cap H_R(\ell_{m+1})$.
\end{lemma}

\begin{figure}[htb]
  \centering
  \includegraphics{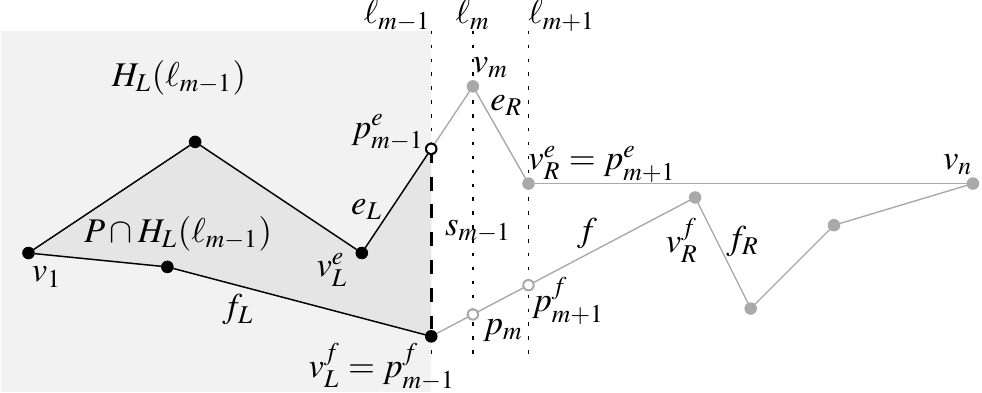}
  \caption{Example for the proof of Lemma~\ref{lem:2k+5always}. 
    The gray shaded area indicates $H_L(\ell_{m-1})$, where $P\cap H_L(\ell_{m-1})$ is shaded slightly darker.
    The bold dashed edge depicts the artificial edge~$s_{m-1}$.}
  \label{fig:2k+5always}
\end{figure}

\begin{proof}
Recall that there are $k+2$ edges of~$P$ contained in $P\cap H_L(\ell_{m-1})$.
Adding $s_{m-1}=P\cap\ell_{m-1}$ to $P\cap H_L(\ell_{m-1})$ as a $(k+3)$rd edge results in a polygon~$L$ that can be illuminated by a \mbox{$k$-modem} placed at any of its vertices, by Lemma~\ref{lem:vertex}.
As $p^e_{m-1}$ and $p^f_{m-1}$ are vertices of~$L$, both points are valid \mbox{$k$-modem} positions for $P\cap H_L(\ell_{m-1})$.
See \figurename~\ref{fig:2k+5always} for an example.

Similarly, the $(k+3)$-gon~$R$, composed of $P\cap H_R(\ell_{m+1})$ and $s_{m+1}=P\cap\ell_{m+1}$, can be illuminated by a \mbox{$k$-modem} placed at any of its vertices, by Lemma~\ref{lem:vertex}.
Hence, $p^e_{m+1}$ and $p^f_{m+1}$ are valid \mbox{$k$-modem} positions for $P\cap H_R(\ell_{m+1})$.

Recall that we can slightly perturb an end point of both~$s_{m-1}$ and~$s_{m+1}$ to maintain the convention of no two vertices sharing a common $x$-coordinate.
%\todo{add sth. like ''Note that here the monotonicity of $P$ is essential to guarantee the complete illumination of $P$''?}
\end{proof}

In the following, we first consider the case of~$k$ being even and then the slightly more involved case of~$k$ being odd.
Before proving Lemma~\ref{lem:2k+5even}, we start with a conditional result for \mbox{$k$-modems} on~$\ell_m$.

\begin{lemma}\label{lem:2k+5subeven}
Let $k\geq0$ be even and let~$P$ be an $x$-monotone $(2k+5)$-gon.

The left part $P\cap H_L(\ell_{m})$ of~$P$ can be illuminated by a \mbox{$k$-modem} placed
at~$v_m$ if the supporting line of $e_L$ does not intersect $f$ to the left of $\ell_m$ and
at~$p_m$ if the supporting line of $f$ does not intersect $e_L$.
Further, at least one of $v_m$ and $p_m$ is a valid \mbox{$k$-modem} position for $P\cap H_L(\ell_{m})$ by these conditions.

The right part $P\cap H_R(\ell_{m})$ of~$P$ can be illuminated by a \mbox{$k$-modem} placed
at $v_m$ if the supporting line of $e_R$ does not intersect $f$ to the right of $\ell_m$ and
at $p_m$ if the supporting line of $f$ does not intersect $e_R$.
Further, at least one of $v_m$ and $p_m$ is a valid \mbox{$k$-modem} position for $P\cap H_R(\ell_{m})$ by these conditions.
\end{lemma}

\begin{proof}
First we consider $v_m$.
By Observation~\ref{obs:simplesplit}, each ray in ${\cal R}(v_m)\cap( P\cap H_L(\ell_{m}) )$ as well as each ray in ${\cal R}(v_m)\cap( P\cap H_R(\ell_{m}) )$ crosses at most $k+3$ edges.
Further, each ray in ${\cal R}(v_m)$ crosses neither~$e_L$ nor~$e_R$ and thus crosses at most $k+2$ edges.
As $k+2$ is even, by Observation~\ref{obs:rayio}, each ray in ${\cal R}^i(v_m)$ crosses at most $k+1$ edges.
If the supporting line of $e_L$ does not intersect $f$ to the left of $\ell_m$, then no ray in ${\cal R}^o(v_m)\cap( P\cap H_L(\ell_{m}) )$ crosses~$f$.
Hence, by Observation~\ref{obs:rayk+1}, $v_m$ is a valid \mbox{$k$-modem} position for $P\cap H_L(\ell_{m})$ in this case.
If the supporting line of $e_R$ does not intersect $f$ to the right of $\ell_m$, then no ray in ${\cal R}^o(v_m)\cap( P\cap H_R(\ell_{m}) )$ crosses~$f$.
Hence, by Observation~\ref{obs:rayk+1}, $v_m$ is a valid \mbox{$k$-modem} position for $P\cap H_R(\ell_{m})$ in this case.

Now we consider $p_m$.
By Observation~\ref{obs:simplesplit} and because no ray in ${\cal R}(p_m)$ crosses~$f$ (which is an edge contained in both $P\cap H_L(\ell_{m})$ and $P\cap H_R(\ell_{m})$), each ray in ${\cal R}(p_m)$ crosses at most $k+2$ edges.
As $k+2$ is even, by Observation~\ref{obs:rayio}, each ray in ${\cal R}^i(p_m)$ intersects at most $k+1$ edges.
If the supporting line of~$f$ does not intersect~$e_L$, then no ray in ${\cal R}^o(p_m)$ crosses~$e_L$.
Hence, by Observation~\ref{obs:rayk+1}, $p_m$ is a valid \mbox{$k$-modem} position for $P\cap H_L(\ell_{m})$ in this case.
If the supporting line of~$f$ does not intersect~$e_R$, then no ray in ${\cal R}^o(p_m)$ crosses~$e_R$.
Hence, by Observation~\ref{obs:rayk+1}, $p_m$ is a valid \mbox{$k$-modem} position for $P\cap H_R(\ell_{m})$ in this case.

Note that the supporting line of~$f$ can intersect~$e_L$ only if the supporting line of~$e_L$ does not intersect~$f$ (to the left of $\ell_m$), and vice versa.
Thus, at least one of~$v_m$ and~$p_m$ is a valid \mbox{$k$-modem} position for $P\cap H_L(\ell_{m})$.
Similarly, the supporting line of~$f$ can intersect~$e_R$ only if the supporting line of~$e_R$ does not intersect~$f$ (to the right of~$\ell_m$), and vice versa.
Thus, at least one of $v_m$ and $p_m$ is a valid \mbox{$k$-modem} position for $P\cap H_R(\ell_{m})$.
\end{proof}

We combine the previous results to prove Lemma~\ref{lem:2k+5even}.

\begin{proof}[Proof of Lemma~\ref{lem:2k+5even}]
For each part $P\cap H_L(\ell_{m})$ and $P\cap H_R(\ell_{m})$ of $P$, at least one of $v_m$ and $p_m$ is a valid \mbox{$k$-modem} position by Lemma~\ref{lem:2k+5subeven}.
If at least one of these positions is a valid \mbox{$k$-modem} position for both $P\cap H_L(\ell_{m})$ and $P\cap H_R(\ell_{m})$, then the lemma is proven.

Thus assume without loss of generality that $v_m$ is not a valid \mbox{$k$-modem} position for $P\cap H_R(\ell_{m})$ and that $p_m$ is not a valid \mbox{$k$-modem} position for $P\cap H_L(\ell_{m})$.
(See \figurename~\ref{fig:evensplit}.) 
We prove that $p^f_{m-1}$ is a valid \mbox{$k$-modem} position for $P$ in this case.
Note that this case implies that $p^f_{m-1}=v_{m-1}=v^f_L$.
By symmetry, $p^f_{m+1}=v_{m+1}=v^f_R$ is a valid \mbox{$k$-modem} position for $P$ in the mirrored case ($v_m$ is not a valid \mbox{$k$-modem} position for $P\cap H_L(\ell_{m})$ and $p_m$ is not a valid \mbox{$k$-modem} position for $P\cap H_R(\ell_{m})$).

\begin{figure}[htb]
  \centering
  \includegraphics{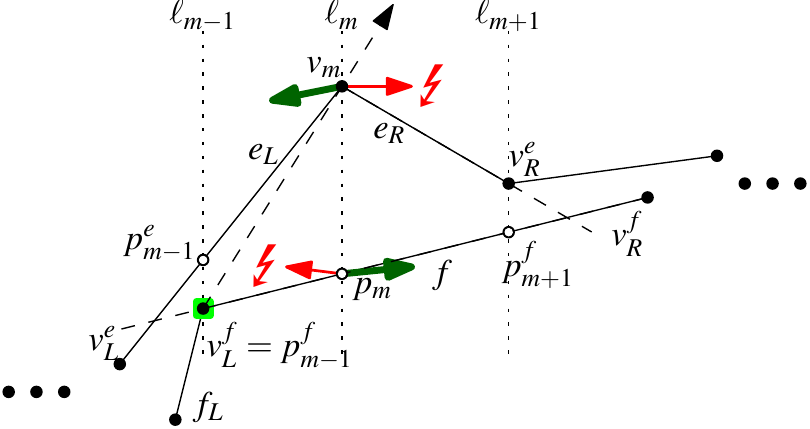}
  \caption{Example for $k$ even. By assumption, $v_m$ is not a valid \mbox{$k$-modem} position for $P\cap H_R(\ell_{m})$ and $p_m$ is not a valid \mbox{$k$-modem} position for $P\cap H_L(\ell_{m})$.}
  \label{fig:evensplit}
\end{figure}

By Lemma~\ref{lem:2k+5subeven}, if~$v_m$ is not a valid \mbox{$k$-modem} position for $P\cap H_R(\ell_{m})$, then the supporting line of~$e_R$ intersects~$f$ to the right of~$\ell_m$.
Further, if~$p_m$ is not a valid \mbox{$k$-modem} position for $P\cap H_L(\ell_{m})$, then the supporting line of~$f$ intersects~$e_L$.
Observe that this implies that $v^e_R$ lies to the right of the ray from $p^f_{m-1}$ towards (and through) $v_m$.
%(See again \figurename~\ref{fig:evensplit}.) 

Therefore, each ray $r\in( {\cal R}(p^f_{m-1})\cap H_R(\ell_{m-1}) )$ that is crossing~$e_L$ cannot cross~$e_R$.
Further, if~$r$ is crossing~$e_L$, then~$r$ is in ${\cal R}^i(p^f_{m-1})$.
(Note that~$e_L$ is the only edge in $P\cap H_L(\ell_{m})$ that can be crossed by~$r$.
No ray~$r$ crosses~$f$.)

The point $p^f_{m-1}$ is a valid \mbox{$k$-modem} position for $P\cap H_L(\ell_{m-1})$ by Lemma~\ref{lem:2k+5always}.
Observe that we split~$P$ at~$\ell_{m-1}$.
By Observation~\ref{obs:simplesplit}, for $P\cap H_R(\ell_{m-1})$ there are in total $k+4$ edges to be considered (including~$f$ and (a part of)~$e_L$).
As no ray in ${\cal R}(p^f_{m-1})$ crosses~$f$, $k+3$ edges remain.

Each ray in $( {\cal R}^o(p^f_{m-1})\cap H_R(\ell_{m-1}) )$ crosses neither~$e_L$ nor~$e_R$, leaving at most $k+1$ edges to cross.
Consider each $r^i\in( {\cal R}^i(p^f_{m-1})\cap H_R(\ell_{m-1}) )$.
Ray~$r^i$ crosses at most one of~$e_L$ and~$e_R$, leaving at most $k+2$ edges to cross.
As $k+2$ is even, $r_i$ crosses at most $k+1$ edges by Observation~\ref{obs:rayio}.
Hence, all rays in ${\cal R}(p^f_{m-1})$ cross at most $k+1$ edges, implying that by Observation~\ref{obs:rayk+1} $p^f_{m-1}$ is a valid \mbox{$k$-modem} position for~$P$.
\end{proof}

With this proof, ``half of'' Theorem~\ref{thm:2k+5} is proven.
The other ``half'' is shown in the next two proofs, which have of the same structure as the two preceeding ones, but need slightly more involved arguments.

\begin{lemma}\label{lem:2k+5subodd}
Let $k\geq1$ be odd and let~$P$ be an $x$-monotone $(2k+5)$-gon.

The left part $P\cap H_L(\ell_{m})$ of~$P$ can be illuminated by a \mbox{$k$-modem} placed
at $v_m$ if $v^e_L$ is reflex or $v^f_L$ is convex and $v^f_L\neq v_1$, and by a \mbox{$k$-modem} placed
at $p_m$ if $v^f_L$ is reflex or $v^e_L$ is convex and $v^e_L\neq v_1$.
Further, at least one of $v_m$ and $p_m$ is a valid \mbox{$k$-modem} position for $P\cap H_L(\ell_{m})$ by these conditions.
If $p_m$ is a valid \mbox{$k$-modem} position for $P\cap H_L(\ell_{m})$ by these conditions and $v^e_L$ is to the left of the ray from $p^f_{m+1}$ towards (and through) $v_m$, then $p^f_{m+1}$ is a valid \mbox{$k$-modem} position for $P\cap H_L(\ell_{m+1})$.

The right part $P\cap H_R(\ell_{m})$ of~$P$ can be illuminated by a \mbox{$k$-modem} placed
at $v_m$ if $v^e_R$ is reflex or $v^f_R$ is convex and $v^f_R\neq v_{2k+5}$, and by a \mbox{$k$-modem} placed
at $p_m$ if $v^f_R$ is reflex or $v^e_R$ is convex and $v^e_R\neq v_{2k+5}$.
Futher, least one of $v_m$ and $p_m$ is a valid \mbox{$k$-modem} position for $P\cap H_R(\ell_{m})$ by these conditions.
If $p_m$ is a valid \mbox{$k$-modem} position for $P\cap H_R(\ell_{m})$ by these conditions and $v^e_R$ is to the right of the ray from $p^f_{m-1}$ towards (and through) $v_m$, then $p^f_{m-1}$ is a valid \mbox{$k$-modem} position for $P\cap H_R(\ell_{m-1})$.
\end{lemma}

\begin{proof}
First we consider $v_m$.
Each ray in ${\cal R}(v_m)$ crosses neither~$e_L$ nor~$e_R$ and thus, by Observation~\ref{obs:simplesplit}, crosses at most $k+2$ edges.
As $k+2$ is odd, by Observation~\ref{obs:rayio}, each ray in ${\cal R}^o(v_m)$ crosses at most $k+1$ edges.
If $v^e_L$ is reflex, then no ray in ${\cal R}^i(v_m)$ crosses the other edge (besides~$e_L$) incident to~$v^e_L$.
If~$v^f_L$ is convex and not the first vertex of~$P$, then each ray in ${\cal R}^i(v_m)$ crosses at most one out of~$f$ and~$f_L$.
In both cases, by Observation~\ref{obs:rayk+1}, $v_m$ is a valid \mbox{$k$-modem} position for $P\cap H_L(\ell_{m})$.

% ---------------------------------------------
% --- repetitive part, second half of lemma ---
% If $v^e_R$ is reflex, then no ray in ${\cal R}^i(v_m)$ crosses the other edge (besides~$e_R$) incident to~$v^e_R$.
% If~$v^f_R$ is convex and not the last vertex of~$P$, then each ray in ${\cal R}^i(v_m)$ crosses at most one out of~$f$ and~$f_R$.
% In both cases, by Observation~\ref{obs:rayk+1}, $v_m$ is a valid \mbox{$k$-modem} position for $P\cap H_R(\ell_{m})$.
% \smallskip
% ---------------------------------------------

Now we consider $p_m$, $p^f_{m-1}$, and $p^f_{m+1}$.
No ray in ${\cal R}(p_m)$, ${\cal R}(p^f_{m-1})$, or ${\cal R}(p^f_{m+1})$ crosses~$f$.
Thus, by Observation~\ref{obs:simplesplit}, each ray in ${\cal R}(p_m)$ crosses at most $k+2$ edges, and each ray in ${\cal R}(p^f_{m-1})\cap H_R(\ell_{m-1})$ or ${\cal R}(p^f_{m+1})\cap H_L(\ell_{m+1})$ crosses at most $k+3$ edges.
Each ray in ${\cal R}^o(p^f_{m-1})\cap H_R(\ell_{m-1})$ cannot cross~$e_L$ and likewise, each ray in ${\cal R}^o(p^f_{m+1})\cap H_L(\ell_{m+1})$ cannot cross~$e_r$.
As $k+2$ is odd, each ray in ${\cal R}^o(p_m)$, ${\cal R}^o(p^f_{m-1})\cap H_R(\ell_{m-1})$, and ${\cal R}^o(p^f_{m+1})\cap H_L(\ell_{m+1})$ cross at most $k+1$ edges by Observation~\ref{obs:rayio}.

% ---------------------------------------------
% \smallskip
% ---------------------------------------------

If $v^f_L$ is reflex, then no ray in ${\cal R}^i(p_m)$ and ${\cal R}^i(p^f_{m+1})$ crosses~$f_L$.
If, in addition, $v^e_L$ is to the left of the ray from $p^f_{m+1}$ through $v_m$, then each ray in ${\cal R}^i(p^f_{m+1})$ that crosses~$e_R$ does not cross~$e_L$.
If $v^e_L$ is convex and not the first vertex of~$P$, then each ray in ${\cal R}^i(p_m)$ crosses at most one out of the two edges incident to $v^e_L$.
If in addition, $v^e_L$ is to the left of the ray from $p^f_{m+1}$ through $v_m$, then each ray in ${\cal R}^i(p^f_{m+1})$ either crosses~$e_R$ but none of the two edges incident to~$v^e_L$ or the ray does not cross~$e_R$ and crosses at most one out of the two edges incident to~$v^e_L$.
Summing up, if $v^f_L$ is reflex or if $v^e_L\neq v_1$ is convex, then each ray in ${\cal R}(p_{m})$ crosses at most $k+1$ edges and thus, $p_m$ is a valid \mbox{$k$-modem} position for $P\cap H_L(\ell_{m})$ by Observation~\ref{obs:rayk+1}.
If, in addition, $v^e_L$ is to the left of the ray from $p^f_{m+1}$ through $v_m$, then each ray in ${\cal R}(p^f_{m+1})\cap H_L(\ell_{m+1})$ crosses at most $k+1$ edges and thus, $p^f_{m+1}$ is a valid \mbox{$k$-modem} position for $P\cap H_L(\ell_{m+1})$ by Observation~\ref{obs:rayk+1}.

% ---------------------------------------------
% --- repetitive part, second half of lemma ---
% \smallskip
% If $v^f_R$ is reflex, then no ray in ${\cal R}^i(p_m)$ and ${\cal R}^i(p^f_{m-1})$ crosses~$f_R$.
% If, in addition, $v^e_R$ is to the right of the ray from $p^f_{m-1}$ through $v_m$, then each ray in ${\cal R}^i(p^f_{m-1})$ that crosses~$e_L$ does not cross~$e_R$.
% %
% If $v^e_R$ is convex and not the last vertex of~$P$, then each ray in ${\cal R}^i(p_m)$ crosses at most one out of the two edges incident to~$v^e_R$.
% If in addition, $v^e_R$ is to the right of the ray from $p^f_{m-1}$ through $v_m$, then each ray in ${\cal R}^i(p^f_{m-1})$ either crosses~$e_L$ but none of the two edges incident to~$v^e_R$ or the ray does not cross~$e_L$ and crosses at most one out of the two edges incident to~$v^e_R$.
% %
% Summing up, if $v^f_R$ is reflex or if $v^e_R\neq v_{2k+5}$ is convex, then each ray in ${\cal R}(p_{m})$ crosses at most $k+1$ edges and thus, $p_m$ is a valid \mbox{$k$-modem} position for $P\cap H_R(\ell_{m})$ by Observation~\ref{obs:rayk+1}.
% %
% If, in addition, $v^e_R$ is to the right of the ray from $p^f_{m-1}$ through $v_m$, then each ray in ${\cal R}(p^f_{m-1})\cap H_R(\ell_{m-1})$ crosses at most $k+1$ edges and thus, $p^f_{m-1}$ is a valid \mbox{$k$-modem} position for $P\cap H_R(\ell_{m-1})$ by Observation~\ref{obs:rayk+1}.
% \medskip
% ---------------------------------------------

As $v^e_L$ as well as $v^f_L$ can either be convex or reflex and not both can be the first vertex of~$P$, at least one out of~$v_m$ and~$p_m$ is a valid \mbox{$k$-modem} position for~$P\cap H_L(\ell_{m})$ by the conditions of the lemma.

\smallskip

Altogether, this proves the first half of the lemma. The second half follows by symmetric arguments.  
%
% ---------------------------------------------
% --- repetitive part, second half of lemma ---
% Similarly, not both~$v^e_R$ and~$v^f_R$ can be the last vertex of~$P$ (and both are either convex or reflex).
% Hence, at least one out of~$v_m$ and~$p_m$ is a valid \mbox{$k$-modem} position for~$P\cap H_R(\ell_{m})$.
% ---------------------------------------------
\end{proof}

We now prove Lemma~\ref{lem:2k+5odd} which concludes the proof for Theorem~\ref{thm:2k+5}.

\begin{proof}[Proof of Lemma~\ref{lem:2k+5odd}]
For each part $P\cap H_L(\ell_{m})$ and $P\cap H_R(\ell_{m})$ of~$P$, at least one of~$v_m$ and~$p_m$ is a valid \mbox{$k$-modem} position by Lemma~\ref{lem:2k+5subodd}.
If at least one of these positions is a valid \mbox{$k$-modem} position for both $P\cap H_L(\ell_{m})$ and $P\cap H_R(\ell_{m})$, then the lemma is proven.

Thus assume without loss of generality that~$v_m$ is not a valid \mbox{$k$-modem} position for $P\cap H_R(\ell_{m})$ and that~$p_m$ is not a valid \mbox{$k$-modem} position for $P\cap H_L(\ell_{m})$.
(See \figurename~\ref{fig:oddsplit}.) 
We prove that at least one out of $p^f_{m-1}$ and $p^e_{m+1}$ is a valid \mbox{$k$-modem} position for~$P$ in this case.
By symmetry, at least one out of $p^f_{m+1}$ and $p^e_{m-1}$ is a valid \mbox{$k$-modem} position for~$P$ in the mirrored case ($v_m$ is not a valid \mbox{$k$-modem} position for $P\cap H_L(\ell_{m})$ and $p_m$ is not a valid \mbox{$k$-modem} position for $P\cap H_R(\ell_{m})$).

\begin{figure}[htb]
  \centering
  \includegraphics{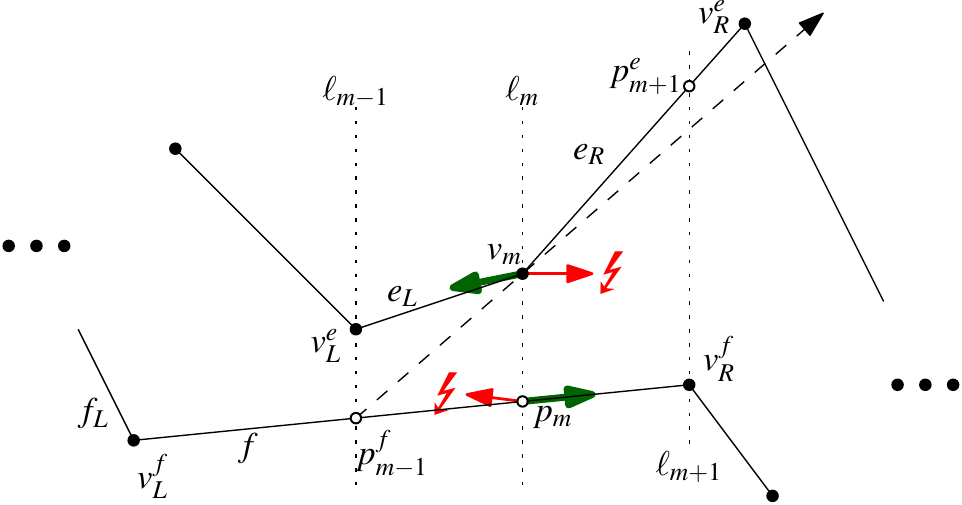}
  \caption{Example for $k$ odd. By assumption, $v_m$ is not a valid \mbox{$k$-modem} position for $P\cap H_R(\ell_{m})$ and $p_m$ is not a valid \mbox{$k$-modem} position for $P\cap H_L(\ell_{m})$.}
  \label{fig:oddsplit}
\end{figure}

If $v^e_R$ is to the right of the ray from $p^f_{m-1}$ through~$v_m$, then, by Lemma~\ref{lem:2k+5subodd}, $p^f_{m-1}$ is a valid \mbox{$k$-modem} position for $P\cap H_R(\ell_{m-1})$, and, by Lemma~\ref{lem:2k+5always}, $p^f_{m-1}$ is a valid \mbox{$k$-modem} position for $P\cap H_L(\ell_{m-1})$.

Thus, assume that $v^e_R$ is to the left of the ray from $p^f_{m-1}$ through~$v_m$.
Observe that this implies that $v_m$ is reflex.
We prove that $p^e_{m+1}$ is a valid \mbox{$k$-modem} position for~$P$ in this case.
By Lemma~\ref{lem:2k+5always} this is true for $P\cap H_R(\ell_{m+1})$.

For $P\cap H_L(\ell_{m+1})$ observe that each ray in ${\cal R}(p^e_{m+1})\cap H_L(\ell_{m+1})$ is crossing at most $k+3$ edges by Observation~\ref{obs:simplesplit} (and because no ray in ${\cal R}(p^e_{m+1})$ crosses~$e_R$).
We distinguish two cases:
\begin{enumerate}
\item \textbf{$v^e_L$ is reflex}:\\
Each ray in ${\cal R}^o(p^e_{m+1})\cap H_L(\ell_{m+1})$ crosses at most one of the two edges incident to~$v^e_L$.
As $k+2$ is odd, each ray in ${\cal R}^o(p^e_{m+1})\cap H_L(\ell_{m+1})$ crosses at most $k+1$ edges by Observation~\ref{obs:rayio}.
Each ray in ${\cal R}^i(p^e_{m+1})\cap H_L(\ell_{m+1})$ crosses neither of the two edges incident to~$v^e_L$ and thus, crosses at most $k+1$ edges.
\item \textbf{$v^e_L$ is convex}:\\
In this case $v^e_L=v_1$ and $v^f_L\neq v_1$ is convex as otherwise, by Lemma~\ref{lem:2k+5subodd}, $p_m$ would be a valid \mbox{$k$-modem} position for $P\cap H_L(\ell_{m})$, which would contradict our assumption.
Hence, each ray in ${\cal R}(p^e_{m+1})\cap H_L(\ell_{m+1})$ crosses at most one out of~$f$ and~$f_L$, leaving at most $k+2$ edges to cross.
As $k+2$ is odd, each ray in ${\cal R}^o(p^e_{m+1})\cap H_L(\ell_{m+1})$ crosses at most $k+1$ edges by Observation~\ref{obs:rayio}.
Further, no ray in ${\cal R}^i(p^e_{m+1})\cap H_L(\ell_{m+1})$ crosses~$e_L$.
\end{enumerate}
Therefore, $p^e_{m+1}$ is a valid \mbox{$k$-modem} position for~$P$ in both cases.
To summarize, we proved that if neither~$v_m$ nor~$p_m$ is a valid \mbox{$k$-modem} position for~$P$ and $p_m$ is not a valid \mbox{$k$-modem} position for $P\cap H_L(\ell_{m})$, then at least one out of~$p^f_{m-1}$ and~$p^e_{m+1}$ is a valid \mbox{$k$-modem} position for~$P$. 
In the symmetric case, if~$p_m$ is not a valid \mbox{$k$-modem} position for $P\cap H_R(\ell_{m})$, then symmetric arguments prove that at least one out of~$p^e_{m-1}$ and~$p^f_{m+1}$ is a valid \mbox{$k$-modem} position for~$P$.
\end{proof}

\subsection{Improved polygon splitting}
\label{sec:gmono-split}

As mentioned before, we will now provide an ``efficient'' way of splitting a monotone polygon~$P$ into smaller monotone polygons.
Unlike the simple splitting in Observation~\ref{obs:simplesplit}, the resulting parts need to be monotone polygons again, so that a recursive splitting is possible.
Further, the sum over the sizes of the smaller polygons should be as small as possible (hence, ``efficient'' splitting).
And, of course, the splitting has to ensure that if all small monotone polygons are illuminated, then~$P$ is illuminated.

\begin{lemma}[Splitting Lemma]\label{lem:splitting} %(Splitting Lemma)
  Let~$P$ be an $x$-monotone $n$-gon and let $3 \leq i \leq  n-1$.
  $P$ can be split into two $x$-monotone polygons~$P_L$ and~$P_R$, such that:
  \begin{itemize}
  \item $P_L$ has $i$ vertices,
  \item $P_R$ has $n-i+2$ vertices, and
  \item if $P_L$ is illuminated by modems placed in $P_L\cap H_L(\ell_{i-1})$ and $P_R$ is illuminated by modems placed in $P_R\cap H_R(\ell_{i})$, then also~$P$ is illuminated by those modems.
  \end{itemize}
\end{lemma}

\begin{figure}[htb]
  \centering
  \includegraphics[page = 1]{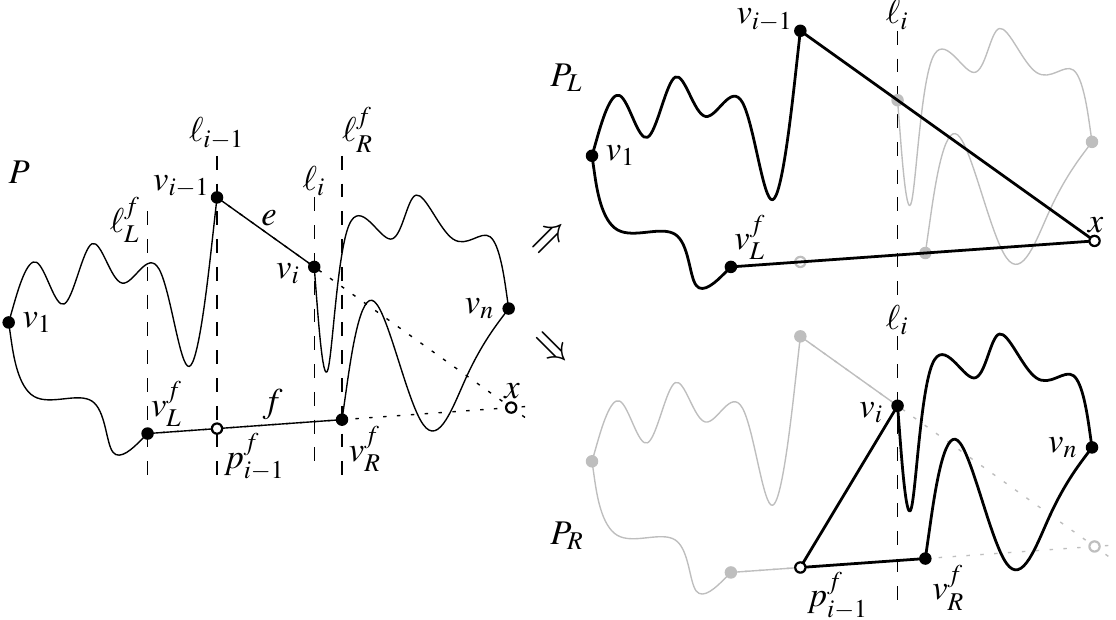}
  \caption{Illustrating the proof of Lemma~\ref{lem:splitting}, with~$x$ to the right of~$v_{i-1}$ and~$v_i$ lies above~$f$.}
  \label{fig:splittinglemma1}
\end{figure}

\begin{proof}
Let~$f$ be the edge of~$P$ intersected by the vertical line~$\ell_{i-1}$.
Without loss of generality, we assume that~$v_{i-1}$ lies above~$f$.
Let~$v^f_L$ and~$v^f_R$ be the left and right end point of~$f$, respectively.
Let~$e$ be the edge of~$P$ having $v_{i-1}$ as its left end point and let~$v^e_R$ be the other end point of~$e$.

Observe that the boundary~$\partial P$ of~$P$ can be partitioned at~$v_1$ and~$v_n$ into an upper polygonal chain~$\partial P^u$ and a lower polygonal chain~$\partial P^l$.
Both chains have~$v_1$ and~$v_n$ as end vertices, but are otherwise disjoint (vertex and edge disjoint).

Since we do not have parallel edges, the supporting lines of the edges~$e$ and~$f$ intersect at a point~$x$.

\begin{figure}[htb]
  \centering
  \includegraphics[page = 2]{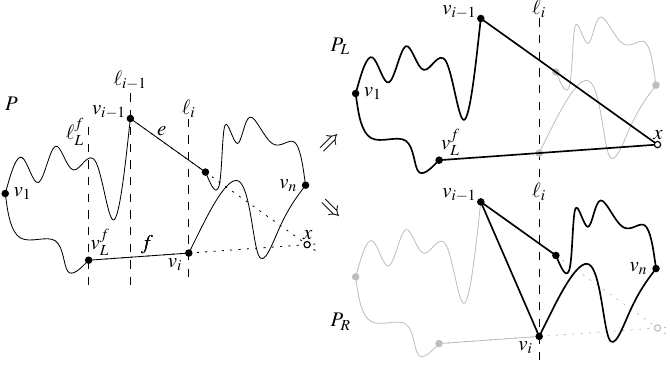}
  \caption{Illustrating the proof of Lemma~\ref{lem:splitting}, with~$x$ to the right of~$v_{i-1}$ and~$v_i$ lies below~$e$.}
  \label{fig:splittinglemma2}
\end{figure}

Assume first that $x$ is to the right of $v_{i-1}$ (see \figurenames~\ref{fig:splittinglemma1},~\ref{fig:splittinglemma2}).
% ---------------------------------------------
% --- repetitive part, second half of lemma ---
%\begin{enumerate}
%\item \textbf{$x$ is to the right of $v_{i-1}$}:\\
% ---------------------------------------------
We construct an upper polygonal chain $\partial P_L^u$ by joining $\partial P^u\cap H_L(\ell_{i-1})$ with the edge from $v_{i-1}$ to $x$.
Let $\ell^f_L$ be the vertical line through~$v^f_L$.
We construct a lower polygonal chain $\partial P_L^l$ by joining $\partial P^l\cap H_L(\ell^f_{L})$ with the edge from~$v^f_L$ to~$x$.
Apart from the common end vertices~$v_1$ and~$x$ these two chains are disjoint.
Hence, joining $\partial P_L^u$ and $\partial P_L^l$ we get the polygonal cycle $\partial P_L$, which is the boundary of the $x$-monotone polygon~$P_L$.
To define the $x$-monotone polygon~$P_R$ we distinguish two cases, depending on whether~$v_i$ lies above~$f$ or below~$e$:
\begin{enumerate}
\item \textbf{$v_i$ lies above~$f$} (\figurename~\ref{fig:splittinglemma1}):
    Let $p^f_{i-1}$ be the intersection of $\ell_{i-1}$ and~$f$.
    The upper polygonal chain $\partial P_R^u$ is the edge from~$p^f_{i-1}$ to~$v_i$ joined with $\partial P^u\cap H_R(\ell_{i})$.
    Let $\ell^f_R$ be the vertical line through~$v^f_R$.
    The lower polygonal chain $\partial P_R^l$ is the edge from~$p^f_{i-1}$ to~$v^f_R$ joined with $\partial P^u\cap H_R(\ell^f_{R})$.
\item \textbf{$v_i$ lies below~$e$} (\figurename~\ref{fig:splittinglemma2}):
    The upper polygonal chain $\partial P_R^u$ is $\partial P^u\cap H_R(\ell_{i-1})$.
    The lower polygonal chain $\partial P_R^l$ is the edge from~$v_{i-1}$ to~$v_i$ joined with $\partial P^u\cap H_R(\ell_{i})$.
\end{enumerate}

Observe that $P_L\cap H_L(\ell_i)=P\cap H_L(\ell_i)$ and $P_R\cap H_R(\ell_i)=P\cap H_R(\ell_i)$.
Therefore, if $P_L$ is illuminated by a modem placed in $P_L\cap H_L(\ell_{i-1})$ and $P_R$ is illuminated by a modem placed in $P_R\cap H_R(\ell_{i})$, then~$P$ is illuminated.

It is easy to see that the resulting polygons~$P_L$ and~$P_R$ are indeed $x$-monotone in all cases.
Further, note that also the convention that no two vertices of a polygon have the same $x$-coordinate is respected by both~$P_L$ and~$P_R$.

Finally,~$P_L$ is an $i$-gon because it contains $i-1$ vertices in $P_L\cap H_L(\ell_{i-1})$ plus one vertex to the right of (and excluding) $\ell_{i-1}$.
Similarly,~$P_R$ is an $(n-i+2)$-gon because it contains $n-i+1$ vertices in $P_R\cap H_R(\ell_{i})$ plus one vertex to the left of (and excluding) $\ell_{i}$.

The situation where $x$ is to the left of $v_{i-1}$ is symmetric and hence can be reasoned in essentially the same way.
\end{proof}

Note that this splitting breaks an $x$-monotone $n$-gon~$P$ into two smaller $x$-monotone polygons,~$P_L$ and~$P_R$, which are not necessarily subpolygons of~$P$.
But, $P\subseteq ( P_L\cup P_R )$.
Further, the restrictions on the placement for modems together with the restrictions on the splitting position~$i$ ensure, which the Splitting Lemma can be applied recursively on the smaller polygons.

\subsection{Illuminating arbitrarily large monotone polygons}
\label{sec:gmono-main}

We proved that a single \mbox{$k$-modem} can illuminate an $x$-monotone $(2k+5)$-gon.
In addition we provided an efficient way to break a ``big'' $x$-monotone $n$-gon into smaller $x$-monotone pieces.
We now combine both results to prove one of our main results presented in the following theorem.

\begin{theorem}\label{thm:gmono-generalbound}
Every $x$-monotone $n$-gon can be illuminated with (at most) $\lceil \frac{n-2}{2k+3}\rceil$~\mbox{$k$-modems}, and there exist $x$-monotone $n$-gons that require at least $\lceil \frac{n-2}{2k+3}\rceil$ \mbox{$k$-modems} to be illuminated.
\end{theorem}

\begin{proof}
For the upper bound let~$P$ be an $x$-monotone $n$-gon.
We iteratively apply the Splitting Lemma (Lemma~\ref{lem:splitting}) to split~$P$ into $t=\left\lceil \frac{n-2}{2k+3} \right\rceil$ $x$-monotone polygons with at most \mbox{$2k\!+\!5$} vertices each, as follows: 
Let~$R_0$ be~$P$.
For $i = 1,\ldots,t-1$, apply Lemma~\ref{lem:splitting} to~$R_{i-1}$ and obtain an $x$-monotone \mbox{$(2k\!+\!5)$-gon} $L_i$ and an $x$-monotone \mbox{$(n\!-\!i(2k\!+\!3))$-gon}~$R_i$.
Let~$L_t$ be the remaining $x$-monotone polygon $R_{t-1}$, which, by definition of $t$, has at most $2k+5$ vertices.
By Lemma~\ref{lem:splitting}, $R_{i-1}$ is illuminated if $L_i$ is illuminated by modems placed anywhere in it but not to the right of its penultimate point and $R_i$ is illuminated by modems placed anywhere in it but not to the left of its second point.
Hence, illuminating each of the obtained $x$-monotone $(2k+5)$-gons $L_1,\ldots,L_t$ with a modem placed not to the left of its second and not to the right of its penultimate point illuminates~$P$.
This is possible for $k\geq 0$ and each $x$-monotone $(2k+5)$-gon by Theorem~\ref{thm:2k+5}.

\begin{figure}[htb]
  \centering
  \includegraphics[page=2]{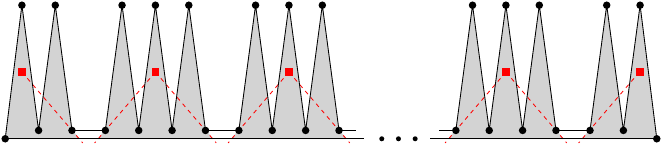} \\
  \vspace{3mm}
  \includegraphics[page=3]{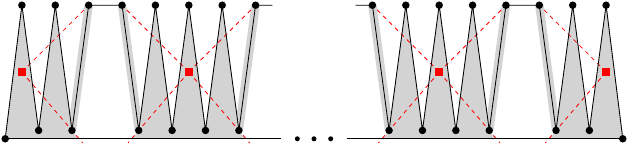} \\
  \caption{Two $x$-monotone $n$-gons that require $\lceil\frac{n-2}{2k+3}\rceil$ \mbox{$k$-modems} each. 
    (Note that the end points of all edges can be slightly perturbed without changing the number of required modems, to respect the convention of no parallel edges.)
    For each of the points in~$W_k$ (depicted as little squares) inside~$P$, the gray region indicates from where it can be illuminated with a \mbox{$2$-modem} (top) or with a \mbox{$3$-modem} (bottom).
    Note that for the lower subfigure, part of these regions are single edges.  }
\label{fig:lower}
\end{figure}

For the lower bound we show how to construct $x$-monotone $n$-gons~$P$ for every value of~$k\geq 0$. 
The construction differs slightly depending on the parity of~$k$.
For each~$k$, the lower chain of~$P$ is a single edge connecting $v_1$ with $v_n$.

For $k$ even, the upper chain of $P$ consists of groups of thin triangular spikes that are separated by quadrilateral valleys. 
The first and the last group of spikes contains $\frac{k+2}{2}$ spikes each, while all other groups contain $k+1$ spikes. 
A sketch of the construction for $k=2$ is shown in \figurename~\ref{fig:lower}~(top).

For $k$ odd, the upper chain of $P$ consists of groups of thin triangular spikes that are separated by quadrilateral towers. 
The first and the last group of spikes contains $\frac{k+1}{2}$ spikes each, while all other groups contain~$k$ spikes. 
A sketch of the construction for $k=3$ is shown in \figurename~\ref{fig:lower}~(bottom).

Observe that each spike needs two vertices, each tower (for $k$ odd) needs three additional vertices, each valley (for $k$ even) needs one additional vertex, and one more vertex is needed to close the polygon. 
Thus, summing up for~$t$ groups of triangular spikes, $P$ has $n=t(2k+3) - 2k$ vertices, independent of whether $k$ is even or odd.

% ------------------------------------------------
% --- this is just for checking for co-authors ---
% to be commented (or removed) for release
%\begin{align*}
%  \mbox{even: }& (k\!+\!2)+1\ +\ (t\!-\!2)\cdot(2(k\!+\!1)+1)\ +\ (k\!+\!2)+1 =\\
%  \mbox{odd: }& (k\!+\!1)+3\ +\ (t\!-\!2)\cdot(2k+3)\ +\ (k\!+\!1) + 1 =\\
%  \mbox{in all cases: }& \ =\ t(2k+3) - 2k \mbox{\textcolor{red}{\quad \ldots\ just for checking}}
%\end{align*}
% ------------------------------------------------

After giving the construction we prove that~$t$ \mbox{$k$-modems} are needed to illuminate~$P$.
We consider a set~$W_k$ of witness points in the interior of~$P$.
One point of~$W_k$ is in the leftmost spike of~$P$.
For each of the \mbox{$t\!-\!2$} groups of $k$~spikes (between two towers or valleys), one point of $W_k$ is in the middle spike.
One last point of $W_k$ is placed in the rightmost spike of~$P$.

Note that the areas from which any two of these points can be illuminated with a \mbox{$k$-modem} are disjoint.
(In the sketch in \figurename~\ref{fig:lower} the points of~$W_k$ are shown as small squares and the areas from which each such witness point can be illuminated is shown shaded.)
Thus, no pair of two points in~$W_k$ can be illuminated by a single \mbox{$k$-modem} and~$P$ requires at least $|W_k|=t$ \mbox{$k$-modems} to be illuminated. 

With $n$ vertices, $P$ can have $t = \left\lfloor \frac{n+2k}{2k+3} \right\rfloor = \left\lceil \frac{n-2}{2k+3} \right\rceil$ groups of spikes and consequently, needs at least that many \mbox{$k$-modems} to be illuminated. %\qed
\end{proof}

Observe that this bound (upper and lower) matches the bound of Chv\'atal's Art Gallery Theorem~\cite{Chvatal}.
The ``watchmen'' there correspond to \mbox{$0$-modems}.
Hence, for $k=0$, the bound $\left\lfloor \frac{n+2k}{2k+3} \right\rfloor$ (from the last paragraph of the proof of Theorem~\ref{thm:gmono-generalbound}) matches the bound~$\left\lfloor\frac{n}{3}\right\rfloor$ from~\cite{Chvatal}.

\section{Illumination of monotone orthogonal polygons}\label{sec:monoortho}

Very often, orthogonal polygons are a sufficiently realistic scenario for placing modems inside buildings in order to cover the interior of the building with wireless reception.
In this section, we give matching lower and upper bounds on the number of \mbox{$k$-modems} needed and required to illuminate monotone orthogonal polygons.

Let~$P$ be an $x$-monotone orthogonal $n$-gon.
Recall that we label the vertices of~$P$, $v_1,\ldots,v_n$ with respect to their $x$-sorted order, such that $v_1$ is the leftmost and $v_n$ is the rightmost vertex of~$P$. In addition, among two vertices with the same $x$-coordinate, the lower vertex (lower $y$-coordinate) gets the higher label. 
For simplicity, we assume that at most two vertices have the same $x$-coordinate.

For validating possible \mbox{$k$-modem} positions, we adopt the respective definitions and observations from Section~\ref{sec:gmono}.
Observe that an orthogonal polygon has (at least) four 'extremal' edges: a topmost, bottommost, leftmost, and rightmost edge, which we denote by $e_t$, $e_b$, $e_l$, and $e_r$, respectively.\footnote{Note that $e_r$ and $e_l$ are unique. If there are more topmost (or bottommost) edges,  let $e_t$ be the leftmost  (or let $e_b$ be the rightmost) among them. }
It is easy to see that for any point $q\in P$ every ray in ${\cal R}(q)$ crosses at most one out of these four extremal edges.

\begin{observation}\label{obs:extremal}
For every orthogonal polygon~$P$ and every point $q\in P$, every ray in ${\cal R}(q)$ crosses at most one out of the four extremal edges of~$P$.
\end{observation}

Using this simple observation, we can prove some first results.

\begin{lemma}\label{lem:k+3}
  Every orthogonal polygon~$P$ with at most $(k+3)$ vertices can be illuminated by a \mbox{$k$-modem} placed anywhere outside~$P$.
\end{lemma}
\begin{proof}
  Any line segment with one end point outside~$P$ and the other end point inside or on the boundary of~$P$ can cross at most one out of the four extremal edges of~$P$ and hence, crosses at most~$k$ edges of~$P$. %\qed
\end{proof}

\begin{lemma}\label{lem:k+4}
Every orthogonal polygon~$P$ with at most $(k+4)$ vertices can be illuminated by a \mbox{$k$-modem} placed anywhere in the interior or on the boundary of~$P$.
\end{lemma}
\begin{proof}
For any point $q\in P$, by Observation~\ref{obs:extremal}, every ray $r\in{\cal R}(q)$ crosses at most $k+1$ edges.
Hence, $q$ is a valid \mbox{$k$-modem} position for~$P$, by Observation~\ref{obs:rayk+1}. %\qed
\end{proof}

\begin{lemma}\label{lem:k+5}
For every orthogonal polygon~$P$ with at most $(k+5)$ vertices there exists a point~$q_l$ on its leftmost edge~$e_l$ and a point~$q_r$ on its rightmost edge~$e_r$, such that~$P$ can be illuminated by a \mbox{$k$-modem} placed at any of~$q_l$ and~$q_r$.
\end{lemma}
\begin{proof}
We prove the case of placing the \mbox{$k$-modem} at~$q_l$.
The other case, for~$q_r$, follows analogously.
If $P$ has at most $k+4$ vertices, then the statement follows by Lemma~\ref{lem:k+4}.
Hence assume that $P$ has $k+5$ vertices. Note that this implies that $k$ is odd (because an orthogonal polygon always has an even number of vertices).

If $e_l$ is incident to a horizontal edge that is not an extremal edge of~$P$, then let this horizontal edge be~$f$.
Otherwise, $e_l$ is incident to the topmost edge~$e_t$ and the bottommost edge~$e_b$ of~$P$.
In this case, let~$f$ be any non-extreme horizontal edge.
(Note that, as $k+5 \geq 4$, there exists at least one such edge.)
For both cases, let $\ell(f)$ be the straight line supporting~$f$ and let $q_l=\ell(f)\cap e_l$.
By Observation~\ref{obs:extremal}, every ray of ${\cal R}(q_l)$ crosses at most $k+2$ edges.
Further, no ray of ${\cal R}(q_l)$ crosses~$f$.
As~$f$ is not an extremal edge of~$P$, every ray of ${\cal R}(q_l)$ crosses at most $k+1$ edges and thus,~$q_l$ is a valid \mbox{$k$-modem} position for~$P$ by Observation~\ref{obs:rayk+1}.
\end{proof}

\begin{figure}[htb]
  \centering
  \includegraphics[page=2]{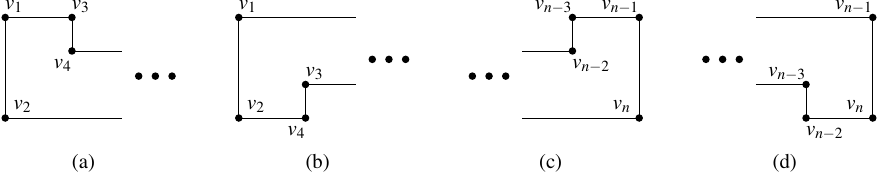}
  \caption{Four examples for stair end polygons:
	  An upper left-sided stair end polygon~(a);
  	  a lower left-sided stair end polygon~(b);
	  an upper right-sided stair end polygon~(c); and 
  	  a lower right-sided stair end polygon~(d).}
  \label{fig:stairend}
\end{figure}

The previous statements are true for general (not necessarily monotone) orthogonal polygons.
The next lemma is only true for a special kind of $x$-monotone orthogonal polygons, which we name upper or lower, left-sided or right-sided \emph{``stair end polygon''}.
An upper left-sided stair end polygon is an $x$-monotone orthogonal $n$-gon that has $v_1v_3$ as an edge.
Likewise, an upper right-sided stair end polygon is an $x$-monotone orthogonal $n$-gon that has $v_{n-3}v_{n-1}$ as an edge.
Further, a lower left- or right-sided stair end polygon is an $x$-monotone orthogonal $n$-gon that has~$v_{2}v_{4}$ or~$v_{n-2}v_{n}$, respectively, as an edge.
See \figurename~\ref{fig:stairend} for examples.

\begin{lemma}\label{lem:k+7}
Let $k\geq 3$ be odd and let~$P$ be a stair end polygon with $(k+7)$ vertices.
If~$P$ is (upper or lower) left-sided, then there exists a point~$q_l$ on its leftmost edge~$e_l$ and if~$P$ is (upper or lower) right-sided then there exists a point~$q_r$ on its rightmost edge~$e_r$, such that~$P$ can be illuminated by a \mbox{$k$-modem} placed on~$q_l$ or~$q_r$, respectively.
\end{lemma}
\begin{proof}
We prove the case of placing the \mbox{$k$-modem} at~$q_r$ for an upper right-sided stair end polygon~$P$.
The lower right-sided case and the upper and lower left-sided cases follow analogously.

Let~$f$ be the horizontal edge incident to~$v_n$, and let~$f'$ be the vertical edge that is the left neighbor of~$f$ on the boundary of~$P$.
Further, let $e_R=v_{n-3}v_{n-1}$, $e_m=v_{n-3}v_{n-2}$, let~$e_L$ be the horizontal edge incident to $v_{n-2}$, and let~$e_L'$ be the vertical edge that is the left neighbor of~$e_L$ on the boundary of~$P$.

Order the horizontal edges of~$P$ by their $y$-coordinates from top to bottom (from $e_t$ to $e_b$).
If two edges have the same $y$-coordinates, let the left one be before the right one.
Let $e_{3t}$ be the third horizontal edge and let $e_{3b}$ be the third-to-last horizontal edge in that order.
We distinguish five cases:

\begin{enumerate}
\item \textbf{$\mathbf{e_{3b}=e_R}$} (see \figurename~\ref{fig:stairendproof1-3}~(a)):\\
Choose $q_r=v_{n-1}$, let $\ell(e_R)$ be the supporting line of~$e_R$, let $r^a\in{\cal R}(q_r)$ be any ray above~$\ell(e_R)$, and let $r^b\in{\cal R}(q_r)$ be any ray below~$\ell(e_R)$.
Note that in this case, $e_R$ is not an extremal edge of~$P$. 
By Observation~\ref{obs:extremal}, $r^a$ crosses at most~$k+4$ edges.
Further, $r^a$ cannot cross the edges~$e_R$, $e_m$, and~$e_L$, leaving at most $k+1$ edges to cross.

As $e_{3b}=e_R$, there are only three horizontal edges, $e_L$, $e_R$, and~$f$, below (or on)~$\ell(e_R)$ and hence, only four vertical edges, $e_L'$, $f'$, $e_m$, and~$e_r$, (in part) below~$\ell(e_R)$. 
Out of this $7$ edges, $r^b$ cannot cross~$e_R$.
Further, $r^b$ can cross at most one edge out of~$e_r$, $f$, and~$f'$, and at most two edges out of~$e_m$, $e_L$, and~$e_L'$.
This leaves also at most $k+1$ edges to cross for~$r^b$.
Therefore, by Observation~\ref{obs:rayk+1}, $q_r$ is a valid \mbox{$k$-modem} position in this case.

\begin{figure}[htb]
  \centering
  \includegraphics[page=1]{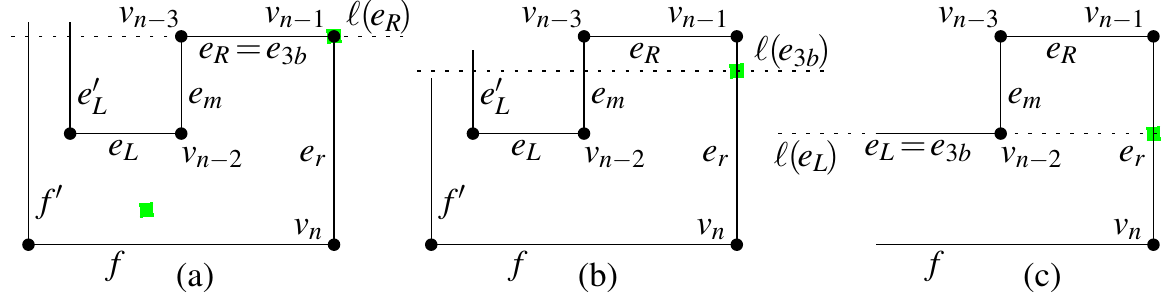}
  \caption{Sketches for the proof of Lemma~\ref{lem:k+7}:
	  Case~${e_{3b}=e_R}$ (a).
	  Case~${e_{3b}}$ is between (not including)~$\mathbf{e_{R}}$ and~$\mathbf{e_{L}}$ (b).
      Case~${e_{3b}=e_L}$ (c).}
  \label{fig:stairendproof1-3}
\end{figure}

\item \textbf{$\mathbf{e_{3b}}$ is between (not including)~$\mathbf{e_{R}}$ and~$\mathbf{e_{L}}$} (see \figurename~\ref{fig:stairendproof1-3}~(b)):\\
Let $\ell(e_{3b})$ be the supporting line of~$e_{3b}$ and choose $q_r=\ell(e_{3b})\cap e_r$.
Let $r^a\in{\cal R}(q_r)$ be any ray above~$\ell(e_{3b})$, and let $r^b\in{\cal R}(q_r)$ be any ray below~$\ell(e_{3b})$.

If~$e_R=e_t$, then $r^a$ either crosses only~$e_R$ or does not cross~$e_R$.
Further, $r^a$ can cross at most one edge out of~$e_r$, $e_b$, $e_l$, and the second horizontal edge.
If~$e_R\neq e_t$, then $r^a$ crosses at most~$k+4$ edges, by Observation~\ref{obs:extremal}, and~$r^a$ crosses at most one edge out of~$e_R$ and~$e_m$.
In both cases, $r^a$ cannot cross the edges~$e_{3b}$ and~$e_L$, leaving at most $k+1$ edges to cross.

Like before, there are only three horizontal edges, $e_L$, $e_{3b}$, and~$f$, below (or on)~$\ell(e_{3b})$ and hence, only four vertical edges, $e_L'$, $f'$, $e_m$, and~$e_r$, (in part) below~$\ell(e_{3b})$. 
Out of this $7$ edges, $r^b$ cannot cross~$e_{3b}$.
Further, $r^b$ can cross at most one edge out of~$e_r$, $f$, and~$f'$, and at most two edges out of~$e_m$, $e_L$, and~$e_L'$.
This leaves also at most $k+1$ edges to cross for~$r^b$.
Therefore, by Observation~\ref{obs:rayk+1}, $q_r$ is a valid \mbox{$k$-modem} position in this case.

\item \textbf{$\mathbf{e_{3b}=e_L}$} (see \figurename~\ref{fig:stairendproof1-3}~(c)):\\
Let $\ell(e_L)$ be the supporting line of~$e_L$ and choose $q_r=\ell(e_L)\cap e_r$.
Let $r^a\in{\cal R}(q_r)$ be any ray above~$\ell(e_L)$, and let $r^b\in{\cal R}(q_r)$ be any ray below~$\ell(e_L)$.

If~$e_R=e_t$, then $r^a$ either crosses only~$e_R$ or does not cross~$e_R$.
Further, $r^a$ can cross at most one edge out of~$e_r$, $e_b$, $e_l$, and the second horizontal edge.
If~$e_R\neq e_t$, then $r^a$ crosses at most~$k+4$ edges, by Observation~\ref{obs:extremal}, and~$r^a$ crosses at most one edge out of~$e_R$ and~$e_m$.
In both cases, $r^a$ cannot cross~$e_L$ and the second-to-last horizontal edge, leaving at most $k+1$ edges to cross.

By Observation~\ref{obs:extremal}, $r^b$ crosses at most~$k+4$ edges.
Further, $r^b$ cannot cross $e_m$, $e_L$, and the second horizontal edge, leaving at most $k+1$ edges to cross.
Therefore, by Observation~\ref{obs:rayk+1}, $q_r$ is a valid \mbox{$k$-modem} position in this case.

\item \textbf{$\mathbf{e_{3b}}$ is below~$\mathbf{e_{L}}$ and $\mathbf{e_{3t}}$ is above~$\mathbf{e_{L}}$} (see \figurename~\ref{fig:stairendproof4-5}~(a)):\\
Let $\ell(e_L)$ be the supporting line of~$e_L$ and choose $q_r=\ell(e_L)\cap e_r$.
Let $r^a\in{\cal R}(q_r)$ be any ray above~$\ell(e_L)$, and let $r^b\in{\cal R}(q_r)$ be any ray below~$\ell(e_L)$.

By Observation~\ref{obs:extremal}, $r^a$ crosses at most~$k+4$ edges.
Further, $r^a$ cannot cross~$e_L$, the second-to-last horizontal edge, and~$e_{3b}$, leaving at most $k+1$ edges to cross.

By Observation~\ref{obs:extremal}, $r^b$ crosses at most~$k+4$ edges.
Further, $r^b$ cannot cross~$e_L$, the second horizontal edge, and~$e_{3t}$, leaving at most $k+1$ edges to cross.
Therefore, by Observation~\ref{obs:rayk+1}, $q_r$ is a valid \mbox{$k$-modem} position in this case.

\begin{figure}[htb]
  \centering
  \includegraphics[page=2]{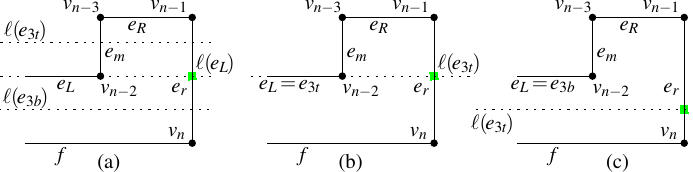}
  \caption{Sketches for the proof of Lemma~\ref{lem:k+7}:
	  Case~${e_{3b}}$ is below~${e_{L}}$ and~${e_{3t}}$ is above~${e_{L}}$ (a).
      Case~${e_{3t}}$ is between (and including)~${e_{L}}$ and~${f}$:
	  ${e_{3t}=e_L}$ (b);
	  ${e_{3t}}$ is between~${e_{L}}$ and~${f}$ (c).
    (The case where~${e_{3t}=f}$ is not depicted.)}
  \label{fig:stairendproof4-5}
\end{figure}

\item \textbf{$\mathbf{e_{3t}}$ is between (and including)~$\mathbf{e_{L}}$ and~$\mathbf{f}$} (see \figurename~\ref{fig:stairendproof4-5}~(b) and~(c)):\\
Let $\ell(e_{3t})$ be the supporting line of~$e_{3t}$ and choose $q_r=\ell(e_{3t})\cap e_r$.
Let $r^a\in{\cal R}(q_r)$ be any ray above~$\ell(e_{3t})$, and let $r^b\in{\cal R}(q_r)$ be any ray below~$\ell(e_{3t})$.

If~$e_R=e_t$, then $r^a$ either crosses only~$e_R$ or does not cross~$e_R$.
Further, $r^a$ can cross at most one edge out of~$e_r$, $e_b$, $e_l$, and the second horizontal edge.
If~$e_R\neq e_t$, then $r^a$ crosses at most~$k+4$ edges, by Observation~\ref{obs:extremal}, and~$r^a$ crosses at most one edge out of~$e_R$ and~$e_m$.
Further, $r^a$ cannot cross $e_{3t}$ and the second-to-last horizontal edge, leaving at most $k+1$ edges to cross.

By Observation~\ref{obs:extremal}, $r^b$ crosses at most~$k+4$ edges.
Further, $r^b$ cannot cross~$e_{3t}$, the second horizontal edge, and~$e_m$, leaving at most $k+1$ edges to cross.
Therefore, by Observation~\ref{obs:rayk+1}, $q_r$ is a valid \mbox{$k$-modem} position in this last case.
\end{enumerate}

As $e_{3b}$ cannot be above~$e_R$ and $e_{3t}$ cannot be below~$f$, this case analysis is exhaustive and proves the claim.
\end{proof}

Like for (general) $x$-monotone polygons, we split a large $x$-monotone orthogonal $n$-gon~$P$ into smaller pieces.
Let $4\leq i\leq n-2$ be even.
We split~$P$ along a vertical line $\ell_i$ through $v_i$ into a left and a right $x$-monotone orthogonal polygon~$P_L$ and~$P_R$, respectively.
Note that $\ell_i$ is the supporting line of the vertical edge of~$P$ with the end points $v_{i-1}$ and $v_i$.
Let~$f$ be the horizontal edge of~$P$ that is crossed by~$\ell_i$ in the point~$p_i$.

Recall that with $H_L(\ell_i)$ ($H_R(\ell_i)$) we denote the left (right) closed half-plane bounded by a vertical line~$\ell_i$.
If the horizontal edge of~$P$ that is incident to~$v_{i-1}$ is completely contained in $H_L(\ell_i)$, then~$P_L$ is $P\cap H_L(\ell_i)$ plus an additional edge~$s_L=v_{i-1}p_i$, and~$P_R$ is $P\cap H_R(\ell_i)$, plus an additional edge~$s_R=v_{i}p_i$.
Otherwise, $P_L$ is $P\cap H_L(\ell_i)$ plus an additional edge~$s_L=v_{i}p_i$, and~$P_R$ is $P\cap H_R(\ell_i)$, plus an additional edge~$s_R=v_{i-1}p_i$.
See \figurename~\ref{fig:orthonaming} for a sketch of the above naming conventions, where the horizontal edge of~$P$ that is incident to~$v_{i-1}$ goes to the right. 

\begin{figure}[htb]
  \centering
  \includegraphics{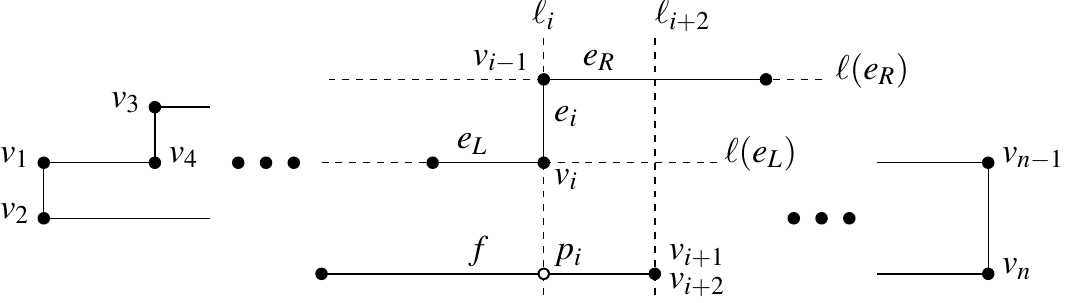}
  \caption{Example to illustrate the naming conventions for an $x$-monotone orthogonal $n$-gon.
    The vertex labeled~$v_{i+1}$ and~$v_{i+2}$ is either vertex~$v_{i+1}$ or vertex~$v_{i+2}$, depending on whether the vertical edge incident to this point goes up or down.
    Note though that in a different example also the right end point of~$e_R$ could be either~$v_{i+1}$ or~$v_{i+2}$.}
  \label{fig:orthonaming}
\end{figure}

It is easy to see that the resulting polygons are subpolygons of~$P$, both $x$-monotone and orthogonal, and that they meet the convention that at most two vertices share a common $x$-coordinate.
As both subpolygons are disjoint except for the common part on~$\ell_i$, illuminating~$P_L$ with modems in (or on the boundary of)~$P_L$ and illuminating~$P_R$ with  modems in (or on the boundary of)~$P_R$ illuminates~$P$.
Further, $P_L$ contains $i$ edges (including~$s_L$) and~$P_R$ contains $n-i+2$ edges (including~$s_R$).
We summarize these observations in the following statement.

\begin{observation}\label{obs:orthosplit}
Let~$P$ be an $x$-monotone orthogonal $n$-gon, $n\geq6$, and let~$\ell_i$ be a vertical line through $v_i$, $4\leq i\leq n-2$, $i$ is even.
Let~$P_L=P\cap H_L(\ell_i)\cup s_L$ and~$P_R=P\cap H_R(\ell_i)\cup s_R$ be subpolygons of~$P$, with~$s_L$ being the rightmost edge of~$P_L$ and~$s_R$ being the leftmost edge of~$P_R$.
\begin{itemize}\vspace*{-0.5ex}
\item $P_L$ is an $x$-monotone orthogonal $i$-gon.
\item $P_R$ is an $x$-monotone orthogonal $(n-i+2)$-gon.
\item If both~$P_L$ and~$P_R$ are illuminated by \mbox{$k$-modems} placed to the left of or on~$\ell_i$ and to the right of or on~$\ell_i$, respectively, then~$P$ is illuminated.
\end{itemize}
\end{observation}

Using the previous results we can prove the following lemma.

\begin{lemma}\label{lem:2k+6}
For every $x$-monotone orthogonal $(2k\!+\!6)$-gon~$P$ there exists a point $q\in P$, such that~$P$ can be illuminated with a \mbox{$k$-modem} placed on~$q$.
\end{lemma}
\begin{proof}
  If~$k$ is even, split~$P$ vertically at $v_{k+4}$ into two $x$-monotone orthogonal $(k\!+\!4)$-gons by Observation~\ref{obs:orthosplit} and let $q=v_{k+4}$.
  By Lemma~\ref{lem:k+4}, both $(k\!+\!4)$-gons are illuminated by a \mbox{$k$-modem} at $q$.

  For odd $k$, split~$P$ vertically at $v_{k+3}$ into one $(k\!+\!3)$-gon~$P_L$ and one $(k\!+\!5)$-gon~$P_R$, both $x$-monotone and orthogonal, by Observation~\ref{obs:orthosplit}.
  Recall that~$s_R$ is the leftmost edge of~$P_R$ and that it is contained in the splitting line. 
  By Lemma~\ref{lem:k+5}, there exists a point~$q$ on $s_R$ where a \mbox{$k$-modem} can be placed to illuminate the $(k\!+\!5)$-gon~$P_R$.
  Further, Lemmas~\ref{lem:k+3} and~\ref{lem:k+4} ensure that the $(k\!+\!3)$-gon~$P_L$ is also illuminated by a \mbox{$k$-modem} at~$q$.

  Hence, whether~$k$ is even or odd, both subpolygons are illuminated and therefore also~$P$, by Observation~\ref{obs:orthosplit}. %\qed
\end{proof}

With this lemma we can prove the following first piece of our main result for monotone orthogonal polygons.

\begin{lemma}\label{lem:orthoupperall}
  Every $x$-monotone orthogonal $n$-gon~$P$ can be illuminated with $\left\lceil \frac{n-2}{2k+4}\right\rceil$ \mbox{$k$-modems}.
\end{lemma}
\begin{proof}
  By Observation~\ref{obs:orthosplit}, we can split~$P$ into a left $x$-monotone orthogonal $(2k\!+\!6)$-gon~$L_1$ and a right $x$-monotone orthogonal $(n\!-\!(2k\!+\!6)\!+\!2)$-gon~$R_1$.
  Recursing on~$R_1$, like in the proof of~Theorem~\ref{thm:gmono-generalbound}, results in $\left\lceil \frac{n-2}{2k+4}\right\rceil$ $x$-monotone orthogonal subpolygons with at most $2k\!+\!6$ vertices each (including~$L_1$ and the rightmost remaining subpolygon).
  By Lemma~\ref{lem:2k+6}, each of these polygons can be illuminated with one \mbox{$k$-modem}.
  By Observation~\ref{obs:orthosplit}, the illumination of all subpolygons implies the illumination of~$P$. %\qed
\end{proof}

Following this upper bound for the number of necessary \mbox{$k$-modems}, we next present a lower bound construction.

\begin{lemma}\label{lem:ortholowerall}
  For even $k$, there exists an $x$-monotone orthogonal $n$-gon requiring $\left\lceil \frac{n-2}{2k+4}\right\rceil$ \mbox{$k$-modems} to be illuminated.
  For odd $k$, there exists an $x$-monotone orthogonal $n$-gon requiring $\left\lceil \frac{n-2}{2k+6}\right\rceil$ \mbox{$k$-modems} to be illuminated.
\end{lemma}

\begin{figure}[htb]
  \centering
  \includegraphics{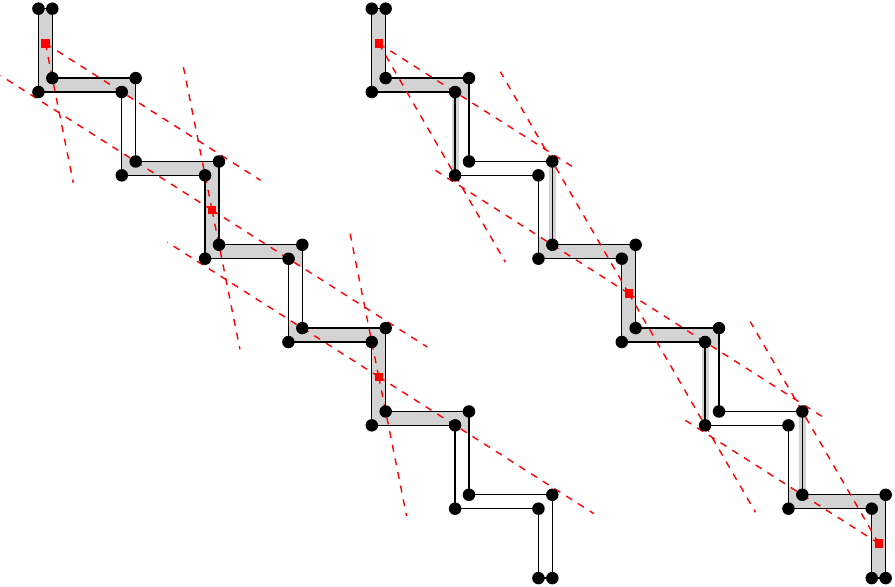}
  \caption{Lower bound construction examples for $x$-monotone orthogonal polygons, with illumination regions for~$k=2$~(left) and~$k=3$~(right).
  }
 \label{fig:ortholowerall}
\end{figure}

\begin{proof}
  The lower bound construction for an $x$-monotone orthogonal $n$-gon~$P$ is sketched in \figurename~\ref{fig:ortholowerall}, left for even~$k$, right for odd~$k$.
  To analyze the construction we place, for each~$k$, a set of~$t$ witness points into~$P$, such that the regions from which each such witness point can be illuminated by one \mbox{$k$-modem} are pairwise disjoint.
  Thus, the number of \mbox{$k$-modems} needed to illuminate~$P$ is at least~$t$.

  For even $k$, we place a witness point in the middle of every ($k\!+\!2$)-th corridor, starting from the leftmost one.
  \figurename~\ref{fig:ortholowerall}~(left) illustrates the set of witness points for~$k=2$.
  The region from which such a witness point can be illuminated with a \mbox{$k$-modem} extends from the corridor containing the point (in both left and right direction) to the next $\frac{k}{2}$ corridors and a small section of the $(\frac{k+2}{2})$-th corridor, which ends before the middle point of that corridor. 
  This way, the regions from which two witness points can be illuminated are disjoint.
  If~$P$ has~$t$ witness points, then~$P$ has at least $1+(k+2)(t-1)$ corridors.
  Hence, $n \geq 2(1+(k+2)(t-1))+2$.
  Thus, in order to illuminate~$P$, the required number of \mbox{$k$-modems} is at least the maximum possible value of~$t$, namely $\lfloor \frac{n+2k}{2k+4} \rfloor = \lceil \frac{n-3}{2k+4} \rceil$.
  Given the fact that $n-3$ is odd and $2k+4$ is even, $P$ requires at least $\lceil \frac{n-2}{2k+4} \rceil$ \mbox{$k$-modems}.

  For odd $k$, we place a witness point in the middle of every ($k\!+\!3$)-th corridor, starting from the leftmost one.
  \figurename~\ref{fig:ortholowerall}~(right) illustrates the set of witness points for~$k=3$.
  The region from which such a witness point can be illuminated with a \mbox{$k$-modem} is the same as for the even case $k-1$, with the addition of one edge of the $(\frac{k+1}{2})$-th corridor (in both left and right direction).
  Thus, this region ends before the midpoint of the $(\frac{k+3}{2})$-th corridor.
  This way, the regions from which two witness points can be illuminated are disjoint.
  By a similar analysis as in the even case, the number of \mbox{$k$-modems} needed to illuminate~$P$ is at least $\lceil \frac{n-2}{2k+6} \rceil$. %\qed
\end{proof}

For the special case of~$k=1$ we can improve the lower bound construction for odd~$k$ to match the bounds of the even case.

\begin{lemma}\label{lem:ortholower1}
  There exists an $x$-monotone orthogonal $n$-gon requiring %$\left\lceil \frac{n-2}{2\cdot 1+4} \right\rceil =$ 
  $\left\lceil \frac{n-2}{6} \right\rceil$ \mbox{$1$-modems} to be illuminated.  
\end{lemma}
\begin{proof}
  Consider the $x$-monotone orthogonal $n$-gon~$P$ sketched in \figurename~\ref{fig:ortholower1}~(middle).
  We partition~$P$ into~$t$ subsets~\mbox{$P\cap G_i$} with ``witness boxes''~$G_i$ (numbered increasingly from top left to bottom right), shown as dotted orthogonal boxes in the figure.
  Further, in each such subset $G_i$ of~$P$ we place two witness points~~$p^r_i$ and~$p^b_i$. 
  We show that~$P$ needs $t= \left\lceil \frac{n-2}{6} \right\rceil$ \mbox{$1$-modems} by proving that in order to illuminate all witness points in~$P$, at least one \mbox{$1$-modem} must be placed per subset~$P\cap G_i$.

\begin{figure}[htb]
  \centering
  \includegraphics{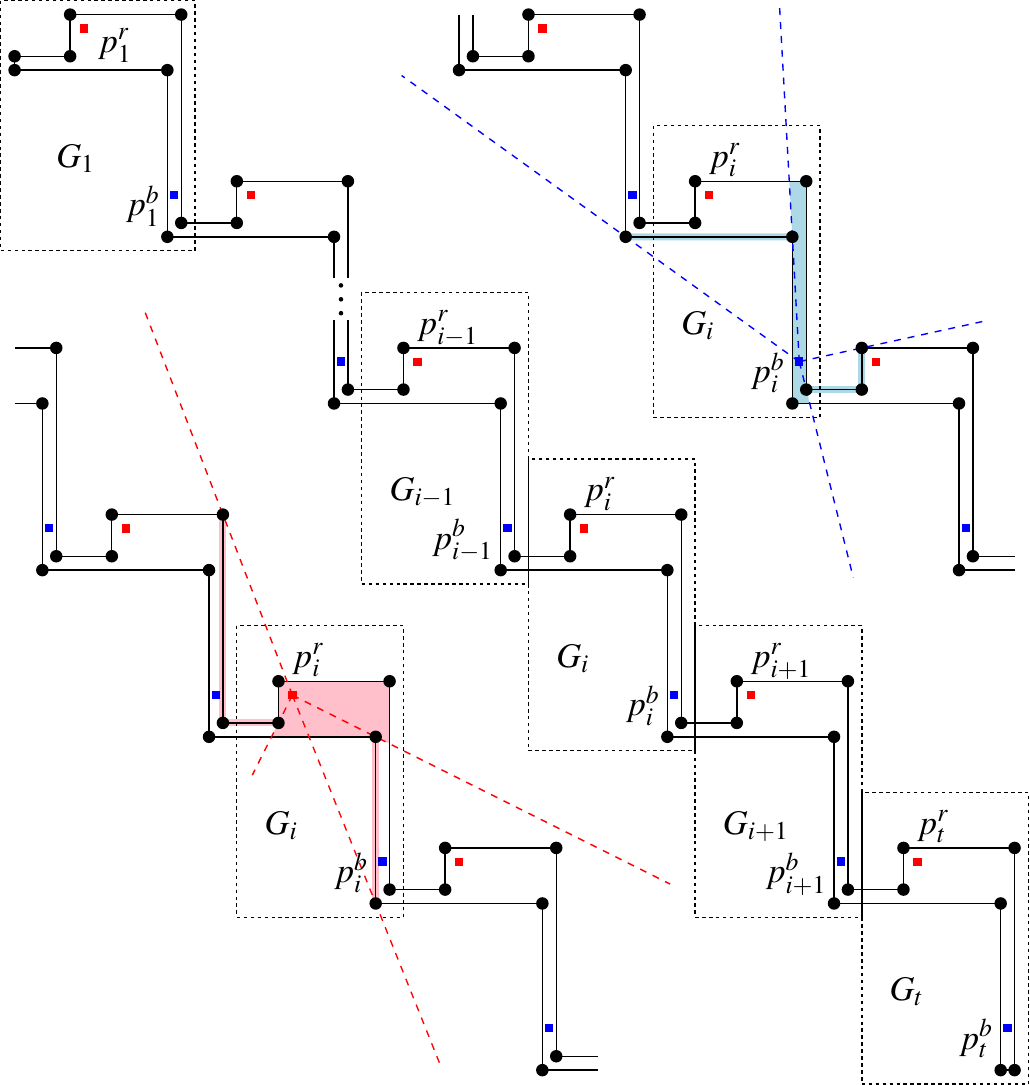}
  \caption{An $x$-monotone orthogonal $n$-gon requiring $\left\lceil \frac{n-2}{6} \right\rceil$ \mbox{$1$-modems} for illumination (middle).
	  \mbox{$1$-modem} region to illuminate the witness point~$p^r_i$ in~$P\cap G_i$ (bottom left).
	  \mbox{$1$-modem} region to illuminate the witness point~$p^b_i$ in~$P\cap G_i$ (top right).
}
 \label{fig:ortholower1}
\end{figure}

\figurename~\ref{fig:ortholower1}~(bottom left) and (top right) depict the regions of~$P$ from which each of the witness points~$p^r_i$ and~$p^b_i$ in one subset~$P\cap G_i$ of~$P$ can be illuminated with a \mbox{$1$-modem}. 
Obviously, it is possible to illuminate all of~$p^r_i$ and~$p^b_i$ using one \mbox{$1$-modem} placed accordingly in each~$P\cap G_i$.
It is also easy to see that no \mbox{$1$-modem} placed in~$P\cap G_i$ can illuminate any of the points~$p^r_j$ and~$p^b_j$ in any other subset~$P\cap G_j$ with~$j < i-1$ or~$j > i+1$.

On one hand, it is possible to illuminate~$p^r_{i+1}$ or~$p^b_{i+1}$ by a \mbox{$1$-modem} placed in~$P\cap G_i$.
On the other hand, no \mbox{$1$-modem} placed in~$P\cap G_i$ can simultaneously illuminate~$p^r_{i+1}$ and~$p^b_{i+1}$.
	
Further, it is possible to illuminate~$p^b_{i-1}$ by a \mbox{$1$-modem} placed in~$P\cap G_i$, but no \mbox{$1$-modem} placed in~$P\cap G_i$ can illuminate~$p^r_{i-1}$.
Finally, a \mbox{$1$-modem} in~$P\cap G_i$ that illuminates~$p^b_{i-1}$ illuminates neither of the points~$p^b_i$, $p^r_{i+1}$, and~$p^b_{i+1}$.
%and any \mbox{$1$-modem} in~$P\cap G_i$ that illuminates~$p^r_{i+1}$ cannot illuminate~$p^b_{i-1}$.

Now consider again the whole $x$-monotone orthogonal $n$-gon~$P$ and assume that~$P$ can be illuminated with less than~$t$ \mbox{$1$-modems}. 
Then there exists a minimum $1 \leq j \leq t$, such that the set $W_j = \{p^r_i, p^b_i : 1 \leq i \leq j\}$ of witness points can be illuminated with less than~$j$ \mbox{$1$-modems}, all located in $(P\cap G_1) \cup \ldots \cup (P\cap G_{j})$. 
Hence, there must be some subset~$P\cap G_h$, $h \leq j$, in which no modem is placed (if a modem is placed in the intersection of two adjacent subsets, then we count it for the left subset).
In the following, we consider the witness points in~$W_j$ according to the order of the subsets of~$P$ and place \mbox{$1$-modems} as needed.
	
In order to illuminate~$p^r_1$ we must place a \mbox{$1$-modem}~$M_1$ in~$P\cap G_1$, implying $j\geq h>1$.
If~$M_1$ does not illuminate~$p^r_2$, then we need a \mbox{$1$-modem} in~$P\cap G_2$ as well. 
Hence, assume that~$M_1$ also illuminates~$p^r_2$ and therefore, does not illuminate~$p^b_2$.
(Note that~$M_1$ illuminates~$p^b_1$ in this case.)

There are two choices for the next \mbox{$1$-modem}.
(1)~We can place a \mbox{$1$-modem}~$M_2$ in~$P\cap G_2$, such that it illuminates~$p^b_2$ and also illuminates~$p^r_3$.
(2)~We can avoid placing a \mbox{$1$-modem} in~$P\cap G_2$ by placing a \mbox{$1$-modem}~$M_2$ in~$P\cap G_3$, such that it also illuminates~$p^b_2$.

For choice~(1) observe that, as long as we keep placing one \mbox{$1$-modem} per subset, the situation stays the same as after placing~$M_1$:
after placing~$i$ \mbox{$1$-modems} for any $1\leq i < h$, all witness points in~$W_{i}$ plus (at most)~$p^r_{i+1}$ are illuminated.
And, in order to illuminate all witness points in~$W_{i+1}$, at least one more \mbox{$1$-modem} is needed.

For choice~(2) assume that each \mbox{$1$-modem}~$M_i$, $1\leq i<h$, has been placed in~$P\cap G_i$ (according to choice~(1)) and~$P\cap G_h$ is the first subset being skipped, i.e., we place the \mbox{$1$-modem}~$M_h$ in~$P\cap G_{h+1}$.
However, in this case~$M_h$ illuminates neither~$p^b_{h+1}$ nor any of the witness points of~$P\cap G_{h+2}$.
Thus, we need a \mbox{$1$-modem}~$M_{h+1}$ in~$P\cap G_{h+2}$ in order to illuminate~$p^b_{h+1}$ which, in turn, does not illuminate~$p^b_{h+2}$.

As long as we now keep placing only one \mbox{$1$-modem} per subset, this situation stays the same as after placing~$M_h$:
after placing~\mbox{$i\!-\!1$} \mbox{$1$-modems} for any $1\leq i\leq j$, all witness points in~$W_{i-1}$ plus (at most)~$p^r_{i}$ are illuminated.
In order to illuminate~$p^b_{i}$ and thus, to completely illuminate~$W_{i}$, at least one more \mbox{$1$-modem} is needed.

To change this situation, we have to place an additional \mbox{$1$-modem} that illuminates~$p^b_{i}$ either in~$P\cap G_{i-1}$ or in~$P\cap G_i$.
In addition to~$W_{i}$, with such a \mbox{$1$-modem} at most one out of~$p^r_{i+1}$ and~$p^b_{i+1}$ can be illuminated.
And this yields the same situation as after either choice~(1) or choice~(2).
Hence, it is not possible to illuminate~$W_j$ with less than~$j$ \mbox{$1$-modems}, for any $1 \leq j \leq t$, contradicting the assumption and proving the lemma.
\end{proof}

Comparing the results in Lemmas~\ref{lem:orthoupperall},~\ref{lem:ortholowerall}, and~\ref{lem:ortholower1}, observe that the bounds are tight if~$k=1$ and if~$k$ is even, but not for odd~$k\geq3$.
In fact, we prove in the next two lemmas that the upper bound for odd~$k$ is indeed lower.

\begin{lemma}\label{lem:2k+8}
For odd~$k\geq3$ and every $x$-monotone orthogonal $(2k+8)$-gon~$P$ there exists a point $q\in P$ such that~$P$ can be illuminated with a \mbox{$k$-modem} placed on~$q$.
\end{lemma}

\begin{figure}[htb]
  \centering
  \includegraphics{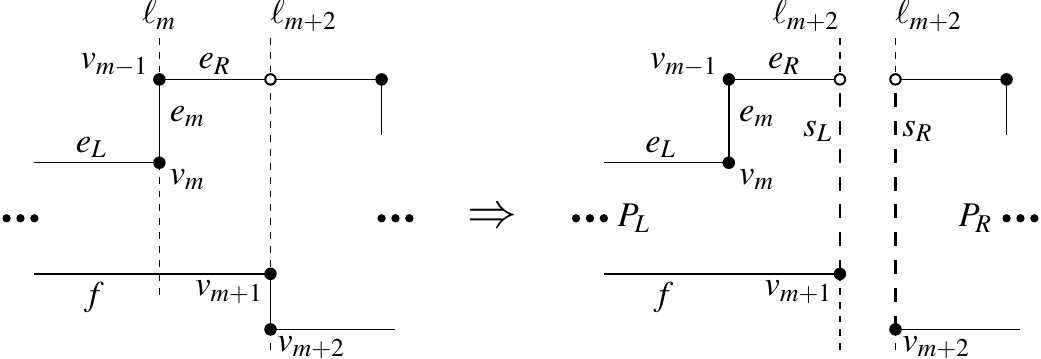}
  \caption{Sketch for proof of Lemma~\ref{lem:2k+8}:
	  ``Middle'' part of~$P$ where the splitting takes place (left).
      The resulting upper right-sided stair end $(k+7)$-gon~$P_L$ and the $x$-monotone orthogonal $(k+3)$-gon~$P_R$. The two bold dashed edges depict~$s_L$ and~$s_R$ (right).}
 \label{fig:ortho2k+8}
\end{figure}

\begin{proof}
Let~$m=k+5$ and let~$\ell_m$ be the vertical line through~$v_m$.
Let~$e_m$ be the vertical edge~$v_{m-1}v_{m}$ of~$P$.
Let~$e_R$ be the horizontal edge of~$P$ that is incident to one end point of~$e_m$ and has its other end point to the right of~$\ell_m$.
Likewise, let~$e_L$ be the horizontal edge of~$P$ that is incident to one end point of~$e_m$ and has its other end point to the left of~$\ell_m$.
Further, let~$f$ be the horizontal edge of~$P$ that is crossed by~$\ell_m$.
See \figurename~\ref{fig:ortho2k+8} for an example. 

We assume that~$e_R$ has~$v_{m-1}$ as one end point and that~$f$ is below~$v_m$ (this case is depicted in \figurename~\ref{fig:ortho2k+8}).
We split~$P$ vertically at~$\ell_{m+2}$ (through~$v_{m+2}$) into one $x$-monotone orthogonal $(k+7)$-gon~$P_L$ and one $x$-monotone orthogonal $(k+3)$-gon~$P_R$ (Observation~\ref{obs:orthosplit}).
Observe that~$P_L$ is an upper right-sided stair end polygon, with~$s_L$ as its rightmost edge, which is contained in the splitting line~$\ell_{m+2}$.
By Lemma~\ref{lem:k+7}, there exists a point~$q$ on $s_L$ where a \mbox{$k$-modem} can be placed to illuminate~$P_L$.
Further, Lemmas~\ref{lem:k+3} and~\ref{lem:k+4} ensure that~$P_R$ is also illuminated by a \mbox{$k$-modem} at~$q$.
Hence, both subpolygons are illuminated and therefore also~$P$, by Observation~\ref{obs:orthosplit}.

If~$e_R$ has~$v_{m}$ as one end point and~$f$ is above~$v_m$, then splitting~$P$ vertically at~$\ell_{m+2}$ (through~$v_{m+2}$) results in a lower right-sided stair end $(k+7)$-gon~$P_L$ (and an $x$-monotone orthogonal $(k+3)$-gon~$P_R$).
In the remaining two cases, $e_R$ has~$v_{m-1}$ as one end point and~$f$ is above~$v_m$, or~$e_R$ has~$v_{m}$ as one end point and~$f$ is below~$v_m$, we split~$P$ vertically at~$\ell_{m-2}$ (through~$v_{m-2}$) into an $x$-monotone orthogonal $(k+3)$-gon~$P_L$ and a lower or upper, respectively, left-sided stair end $(k+7)$-gon~$P_R$.
In all three cases, an analogous argumentation as above proves the lemma. %\qed
\end{proof}

With this lemma we can prove the last missing piece of our main result for monotone orthogonal polygons.

\begin{lemma}\label{lem:orthoupperodd}
For odd~$k\geq3$, every $x$-monotone orthogonal $n$-gon~$P$ can be illuminated with $\left\lceil \frac{n-2}{2k+6}\right\rceil$ \mbox{$k$-modems}.
\end{lemma}
\begin{proof}
By Observation~\ref{obs:orthosplit}, we can split~$P$ into a left $x$-monotone orthogonal $(2k\!+\!8)$-gon~$L_1$ and a right $x$-monotone orthogonal $(n\!-\!(2k\!+\!8)\!+\!2)$-gon~$R_1$.
Like in the proof of~Lemma~\ref{lem:orthoupperall}, recursing on~$R_1$ results in $\left\lceil \frac{n-2}{2k+6}\right\rceil$ $x$-monotone orthogonal subpolygons with at most $2k\!+\!8$ vertices each (including~$L_1$ and the rightmost remaining subpolygon).
By Lemma~\ref{lem:2k+8}, each of these polygons can be illuminated with one \mbox{$k$-modem}.
By Observation~\ref{obs:orthosplit}, the illumination of all subpolygons implies the illumination of~$P$. % %\qed
\end{proof}

We summarize the results for monotone orthogonal polygons from Lemmas~\ref{lem:orthoupperall},~\ref{lem:ortholowerall},~\ref{lem:ortholower1}, and~\ref{lem:orthoupperodd} in the following Theorem.

\begin{theorem}\label{thm:orthoall}
Let~$P$ be an $x$-monotone orthogonal $n$-gon.
For $k=1$ and all even~$k$, $\left\lceil \frac{n-2}{2k+4} \right\rceil$ \mbox{$k$-modems} are always sufficient and sometimes necessary to illuminate~$P$.
For odd~$k\geq3$, $\left\lceil \frac{n-2}{2k+6}\right\rceil$ \mbox{$k$-modems} are always sufficient and sometimes necessary to illuminate~$P$.
\end{theorem}

\section{Conclusion}

Inspired by current wireless networks, we studied a new variant of the classic polygon-illumination problem.
To model the way wireless devices communicate within a building, we allow signals to cross at most a given number~$k$ of walls.

Using as a main tool the Splitting Lemma that allows us to divide a polygon into simpler, overlapping polygons, we gave an upper bound of $\lceil \frac{n-2}{2k+3} \rceil$ on the number of \mbox{$k$-modems} needed to illuminate any given monotone polygon with~$n$ vertices.
We also presented a family of monotone polygons that need at least $\lceil \frac{n-2}{2k+3} \rceil$ \mbox{$k$-modems}, which shows that our upper bound is tight.

Further, we also studied the particular case when the monotone polygons are orthogonal, where we derived similar tight bounds, which differ depending on the parity of~$k$.
For even~$k$, $\left\lceil \frac{n-2}{2k+4} \right\rceil$ \mbox{$k$-modems} are always sufficient and sometimes needed to illuminate a monotone orthogonal $n$-gon.
And for odd~\mbox{$k\geq3$}, $\left\lceil \frac{n-2}{2k+6} \right\rceil$ \mbox{$k$-modems} are always sufficient and sometimes needed to illuminate a monotone orthogonal $n$-gon.
Interestingly, the bounds for the number of \mbox{$k$-modems} for~$k=1$ are the same as for even~$k$.
This is an artifact of the orthogonality and the small constant~$1$.

Let us conclude with the following open problem: 
What is the algorithmic complexity of finding the minimum number of \mbox{$k$-modems} (and their position)
to illuminate a given monotone (orthogonal) polygon?

% \tom{In the following paragraph Ruy's students should be
%   mentioned. Maybe Ruy can rewrite/update this paragraph?}
% %
% The \mbox{$k$-modem} illumination problem for general polygons has proven to be rather challenging. 
% We believe that obtaining tight bounds for this case is a non-trivial problem. 
% At the moment the best lower bounds we have are those we obtained for monotone polygons, and no non-trivial upper bounds are known to us.
% %
% In fact it could be that the lower bounds for general polygons happen to be achieved by monotone polygons as well, similar to the  $\left\lfloor n / 3 \right\rfloor$ lower bound for the classical (``\mbox{$0$-modem}'') polygon illumination problem.

\section*{Acknowledgments}
We thank Clemens Huemer and Jan P{\"o}schko for helpful discussions.

%%%%%%%%%%%%%%%%%%%%%%%%%%%%%%%%%%%%%%%%%%%%%%%%%%%%%%%%%%%%%%%%%%%%%%%%%%%%%%

%---------------------------- Bibliography -------------------------------
%{\footnotesize
%\bibliographystyle{model1b-num-names}
%\bibliography{bichromaticMatchings}
%}

%\vspace{.2in}
%\newpage
%{\small
%\bibliographystyle{abbrv}
%\bibliography{bichromaticMatchings}
%}

\end{document}